\documentclass[11pt]{article}

\usepackage[margin=2.54cm]{geometry}
\usepackage[dvipsnames,usenames]{color}
\usepackage{amsfonts,amsmath,amssymb,amsthm}
\usepackage{amsfonts,amsmath,amssymb}
\usepackage{algorithmic,algorithm}
\usepackage{bm}
\usepackage{times}
\usepackage{verbatim}
\usepackage{graphicx}
\usepackage[pagebackref,letterpaper=true,colorlinks=true,pdfpagemode=none,urlcolor=blue,linkcolor=blue,citecolor=red]{hyperref}
\usepackage{xspace,prettyref}
\usepackage[T1]{fontenc}
\usepackage[round,comma,sort]{natbib}

\let\pref=\prettyref
\newcommand{\savehyperref}[2]{\texorpdfstring{\hyperref[#1]{#2}}{#2}}

\newrefformat{eq}{\savehyperref{#1}{\textup{(\ref*{#1})}}}
\newrefformat{lem}{\savehyperref{#1}{Lemma~\ref*{#1}}}
\newrefformat{def}{\savehyperref{#1}{Definition~\ref*{#1}}}
\newrefformat{thm}{\savehyperref{#1}{Theorem~\ref*{#1}}}
\newrefformat{cor}{\savehyperref{#1}{Corollary~\ref*{#1}}}
\newrefformat{cha}{\savehyperref{#1}{Chapter~\ref*{#1}}}
\newrefformat{sec}{\savehyperref{#1}{Section~\ref*{#1}}}
\newrefformat{app}{\savehyperref{#1}{Appendix~\ref*{#1}}}
\newrefformat{tab}{\savehyperref{#1}{Table~\ref*{#1}}}
\newrefformat{fig}{\savehyperref{#1}{Figure~\ref*{#1}}}
\newrefformat{hyp}{\savehyperref{#1}{Hypothesis~\ref*{#1}}}
\newrefformat{alg}{\savehyperref{#1}{Algorithm~\ref*{#1}}}
\newrefformat{item}{\savehyperref{#1}{Item~\ref*{#1}}}
\newrefformat{step}{\savehyperref{#1}{step~\ref*{#1}}}
\newrefformat{conj}{\savehyperref{#1}{Conjecture~\ref*{#1}}}
\newrefformat{fact}{\savehyperref{#1}{Fact~\ref*{#1}}}
\newrefformat{prop}{\savehyperref{#1}{Proposition~\ref*{#1}}}
\newrefformat{claim}{\savehyperref{#1}{Claim~\ref*{#1}}}

\newtheorem{theorem}{Theorem}[section]
\newtheorem{lemma}[theorem]{Lemma}

\newtheorem{proposition}[theorem]{Proposition}

\newtheorem{definition}{Definition}
\theoremstyle{definition}

\newcommand{\Z}{\mathbb Z}

\newcommand{\eqdef}{\stackrel{\textit{\rm\tiny def}}{=}}

\newcommand{\abs}[1]{\left\lvert #1 \right\rvert}

\newcommand{\pwlmon}[1]{\mathcal{P}^{(#1)}}
\newcommand{\pprice}[1]{\mathcal{P}^{(#1)}}
\newcommand{\peq}[1]{\mathcal{P}^{(#1)}_1}
\newcommand{\palloc}[1]{\mathcal{P}^{(#1)}}

\newcommand{\pallocseq}[1]{\mathcal{P}^{(#1)}_2}
\newcommand{\pallocseqtie}[1]{\mathcal{P}^{(#1)}}
\newcommand{\ppricemon}[1]{\mathcal{P}^{(#1)}_3}
\newcommand{\putilmon}[1]{\mathcal{P}^{(#1)}_4}
\newcommand{\putilcont}[1]{\mathcal{P}^{(#1)}_5}

\newcommand{\price}[2]{p^{(#1)}_{#2}}

\begin{document}

\begin{titlepage}
\thispagestyle{empty}

\title{Sequential Auctions of Identical Items with\\
Budget-Constrained Bidders}

\author{Zhiyi Huang\thanks{University of Pennsylvania. Supported in part by an ONR MURI Grant N000140710907. Part of this work was done while the author was an intern at Microsoft Research, Redmond. Email: {\tt hzhiyi@cis.upenn.edu}.}
\and
Nikhil R. Devanur\thanks{Microsoft Research, Redmond. Email: {\tt nikdev@microsoft.com}.}
\and
David Malec\thanks{University of Wisconsin-Madison. Part of this work was done while the author was an intern at Microsoft Research, Redmond. Email: {\tt  dmalec@cs.wisc.edu}.}}

\date{}

\maketitle

\begin{abstract}
    \thispagestyle{empty}
    In this paper, we study sequential auctions with two budget constrained bidders and any number of identical items. 
    All prior results on such auctions consider only two items. 
 We construct a canonical outcome of the auction that is the only natural equilibrium and is unique under a refinement of subgame perfect equilibria. 
 We show  certain interesting properties of  this equilibrium; for instance, we show that the prices decrease as the auction progresses. This phenomenon has been observed in many experiments and 
 previous theoretic work attributed it to features such as uncertainty in the supply or  risk averse bidders. 
 We show that such features are not needed for this phenomenon and that it arises purely from the most essential features: budget constraints and the sequential nature of the auction.
A little surprisingly we also show that in this equilibrium one agent wins all his items in the beginning and then the other agent wins the rest. The major difficulty in analyzing such sequential auctions has been in understanding how the selling prices of the first few rounds affect the utilities of the agents in the later rounds. We tackle this difficulty by identifying certain key properties of the auction and the proof is via a joint induction on all of them. 
\end{abstract}

\end{titlepage}

\section{Introduction} 
\label{sec:intro}


Currently, there is a rich and well developed theory of auctions, most of it focused on single-round auctions, where all the items are auctioned off at once. 
Yet, many real world instances of auctions for multiple items are sequential in nature and the theory behind such auctions is still in a nascent stage. 
Another commonly occurring real-world aspect of auctions is budget constraints and incorporating budget constraints into traditional auction theory has often been quite challenging. 
The combination of sequential auctions and budget constraints is very natural and occurs commonly. 
This has been well recognized by the economics community. 
A modern-day instance of  this combination is the much studied ad-auction; there have been many theoretical results capturing different aspects of  ad-auctions. 
However there has been little to no theoretical results on the sequential nature of these auctions, even though there have been 
empirical studies that suggest bidders in these auctions do strategize to exhaust their competitors' budgets  
and take advantage of an emptier playing field \citet{ZhouL07}.  
The goal of this paper is to characterize and extend our understanding of sequential auctions with budget constrained bidders.

Our results are based on a simple yet rich model that preserves the most essential features we wish to understand: the sequential nature combined with budget constraints. 
The model is simple because the items are identical, the number of items is fixed and known and agents have complete information.\footnote{Understanding idealized models such as one with complete information case provides a benchmark with which we can compare more realistic settings. }
Yet, the model is rich enough that the equilibrium outcomes display a rather complicated pattern (See Tables \ref{tab:twoitem} and \ref{tab:threeitem}). 
Our results need substantial work and the difficulties we face and the techniques we employ are summarized in  \pref{sec:difficulties}.

Our model is as follows: there are multiple identical items and
two agents interested in obtaining them. Each agent wants to acquire as
many items as possible but is subject to a {\em budget} constraint. The
items are auctioned off {\em sequentially}, that is, the items are
auctioned one after the other, each auction starting after the previous
one is completed.  The agents have {\em complete information}, that is
each agent knows the budget of the other and the total  number of items.
This scenario can be thought of as a $2$-player extensive game. The
question we seek to answer in this paper is, what are the subgame
perfect equilibria of this game. 

The above question is one of the most basic that can be asked about
sequential auctions with budget constrained bidders. The intuition for
what should happen in the above game is also fairly straight forward:
each agent tries to exhaust the budget of the other as soon as possible,
so that he can have the rest of the items for himself at a low price.
This intuition as well as the importance of this question was recognized
as early as \citet{pitchik1986budget}. However so far it has not been
possible to analyze such games with more than two items and these games
have come to be regarded as quite difficult to analyze. The difficulty
arises from the fact that as the number of rounds increases, the effects
of an outcome in the first round on those in the later rounds become
extremely complex. At least part of this difficulty arises due to the
budget constraint as opposed to, say, a unit-demand constraint. (We
present a detailed discussion of these aspects in
\pref{sec:difficulties}).  

In each round we consider a first-price sealed-bid auction. One could
also consider alternatives such as a second-price auction
(See \pref{sec:secondprice} for a comparison) or an ascending price
auction, but for technical reasons a first-price auction is most
appropriate. One issue with 
subgame perfect equilibria as such is that there are many unnatural equilibria, for instance, even for a single-item auction. 
We need to consider a stronger notion of equilibrium in 
order to rule out these unnatural equilibria, and the one we use is a variant of the
trembling hand perfect equilibrium. (See \pref{sec:prelim} for a precise definition.) 
Under this refinement, we show that there is a unique subgame perfect equilibrium of this game, {\em for any number of items}.
(All earlier results only consider two items.) 
This unique equilibrium is indeed the natural equilibrium of the game and we refer to this as the canonical equilibrium.
We show several properties of this equilibrium,  some surprising and some expected. 
\begin{itemize}
\item The number of items won by an agent is approximately proportional to his budget. 
\item The prices are monotonically non-increasing as the auction progresses. 
\item One agent wins all of his items in the beginning and then the other agent wins the remaining items. 
\end{itemize} 

We now discuss each of these properties in more detail. 

\paragraph{Number of Items Won} This is perhaps the least surprising of all the properties. It is natural to expect that the number
of items won is approximately proportional to the budget. However, an interesting contrast surfaces when we compare the number
of items won here with the number of items won in the adaptive clinching
auction of \citet{dobzinski2011multi}. It turns out that
the number of items won in the sequential auction is more equitable, and closer to the proportion of the budget than in the 
adaptive clinching auction. A detailed comparison is presented in
\pref{sec:clinching}. 

\paragraph{Monotonicity of Prices} \citet{pitchik1988perfect} ran lab
experiments of sequential auctions that showed that prices decrease as
the auction progresses. Since then such experiments have been repeated
by others and the monotonic decrease of prices has been reaffirmed
\citet{ashenfelter1989auctions}. Almost all subsequent theoretical
results have tried to capture this phenomenon. Actually the first among
these, e.g., \citet{weber1981multiple} and \citet{milgrom1982theory}
showed the opposite, that prices increase. 
Later, different results attributed the decreasing price phenomenon to different features of the model: \citet{jeitschko1999equilibrium} showed
that decreasing prices could occur due to uncertainty in
the number of items, \citet{black1992systematic} attributed it to
decreasing marginal utilities and \citet{mcafee1993declining} to risk aversion. The game we
consider does not have any of these features: the number of items  is
fixed and known, the marginal utilities are constant and the agents are
risk neutral. Our result shows that the decreasing prices phenomenon
occurs purely from the most essential properties: the budget constraint
and the sequential nature of the game. 

\paragraph{The Order of Sale} As one would expect, the number of items
won by an agent is approximately proportional to his budget. The real
surprise is in the  order in which the items are won. Although the
intuition we mentioned earlier says that an agent tries to exhaust the
other agent's budget, it is conceivable that they win alternately to
maintain the ratio of items won along the way. But the equilibrium
outcome is at the other extreme. The ratio of the budgets not only
determines the number of items each agent wins, but also determines
which agent can force the other one to win all of his items in the
beginning and then win the rest of the items. Suppose that the  budget
of agent 1 is fixed and the budget of agent 2 is monotonically
increasing. The outcome changes as follows. Say agent 2 is winning the
first few items initially. At some point his budget becomes large enough
that he can force agent 1 to win the items first, while keeping the {\em
number}  of items won the same. As his budget increases further, there
is once again a transition point where agent 2 can win one more item,
but now has to move back to winning items first. The pattern then
continues until he can win all the items.  There may be some
discontinuity at the transition points since tie breaking rules come
into effect.

\subsection{Related Work} 
Although many variants of sequential auctions with budget constrained bidders have been considered, they are all limited to just two rounds. 
   The seminal paper of \citet{pitchik1986budget} considered sequential auctions with two (non-identical) items and two budgeted constrained bidders with additive valuation over the items. They characterize the equilibria of this game in the complete information setting and prove certain properties. 
   They followed it up with an experimental validation of the decreasing prices phenomenon in \citet{pitchik1988perfect}. 
   Additional empirical evidence of decreasing prices was provided in \citet{ashenfelter1989auctions}.  
   In \citet{benoit2001multiple} the agents have a common value for 2 non identical items with combinatorial valuations.    
   \citet{rodriguez2009sequential} constructs equilibrium with
        declining prices for sequential auction with multiple agents
        under the assumption of decreasing marginal utility. 
   An isolated example of sequential auctions in the context of ad-auctions is \citet{iyengar2006bidding}. They also consider two rounds and their results are dominating strategies for bidders participating in such auctions. 
   Some examples of sequential auctions in settings other than budget constrained bidders are 
        \citet{weber1981multiple,black1992systematic,mcafee1993declining,jeitschko1999equilibrium,hassidim2011non,paes2011sequential}.   
   Budget constraints have also been considered in the setting of simultaneous (vs. sequential) auction design by
        \citet{dobzinski2011multi,borgs2005multi,bhattacharya2010budget,che1998standard,che2000optimal}. 

Our model is also similar to some of the literature on contests such as \cite{Robson}, but there are crucial differences. 
For instance, the problem considered by \cite{Robson} is as follows: for one item, if the two bidders bid $x$ and $y$, then they {\em both pay their respective bids}  and 
the result is that bidder 1 wins with probability $x^r/(x^r + y^r)$ where $0<r\leq 1$ is some parameter, and 
bidder 2 wins with the remaining probability. Plus there is no {\em sequentiality} in his model. 
These make a huge difference; in his model the equilibrium can be easily solved for and a closed form expression for equilibrium bids is derived. 
In fact, if the items are identical (as in our model), then his solution is simply to divide the budgets equally among all items.
This is very different from the equilibrium behavior in our model as can be seen from Tables \ref{tab:twoitem} and \ref{tab:threeitem}. 
 
\subsection{Techniques and Difficulties} \label{sec:difficulties}

The simplest of all the properties we prove is that the number of items won is approximately proportional to the budgets. Let the total number of rounds be $k$ and the budgets of the agents be $B_1$ and $B_2$. Suppose  agent 1 bids $B_2/k_2$ in every round, for some integer $k_2 \in [0,k]$. Agent 2 can afford to win at most $k_2$ items with this strategy. Therefore agent 1 could ensure $k_1 := \min \{ k - k_2, B_1k_2/B_2 \} $ items.   Optimizing the choice of $k_1$, agent 1 can ensure he wins $k_1 $ items for the highest integer $k_1$ for which $k_1 + k_2 = k$ and 
$B_1/k_1 \geq B_2/k_2$.

The major difficulty comes (naturally) from understanding the interplay
between individual auctions resulting from the shared budget across
them. The central question is what the outcome of the current auction
means to an agent in terms of later rounds.  Intuitively, it seems
sensible that a winning agent always prefers to pay as little as
possible, while a losing agent prefers the winner's payment to be as
high as possible. 
%
Thus, in order to understand the outcome of our auctions, we need to
understand how the utility from a sequence of $k$ auctions
behaves as a function of the remaining budgets.  
Unfortunately the structure of the utility function is rather complicated.  
While an agent's utility from the current round is still linear in the payment for the current
round, there is no guarantee that the role the payment plays in
reducing the budget for future rounds is also linear -- and in fact,
analyzing even cases with few items makes it clear that it is not.  
%
Analyzing small cases, i.e.~$k=1,2,3$ shows that while it does start out simple, adding more
items quickly makes the function very  complicated.  This can be seen
in \pref{tab:twoitem} and \pref{tab:threeitem}.

It is interesting to compare sequential auctions with unit-demand
bidders and budget constrained bidders. Unit-demand bidders only want
one copy of the item, as long as the price is less than their valuation
of the item. In one round of a sequential auction with unit-demand
bidders, only the identity of the winner in that round matters to the
remaining bidders. The subgame only depends on the remaining bidders. This
can be interpreted as bidders having externalities on each other
\citet{paes2011sequential}.  In contrast, in a sequential auction with budget constrained bidders, not only the identity of the winner but also the price paid by her matters to the other bidder. 
The subgame and in particular the utility function depends on the residual budgets.  

%

 First of all, in order to even talk about  {\em the} utility,  you need that the equilibrium is unique, which seems circular. 
We get around this by first defining a {\em canonical outcome} and consider the utilities of the agents in this outcome. 
We identify certain key properties of the canonical outcome such as the monotonicity of prices,
the sequence of wins and the monotonicity and continuity of the utility
function that help us show that this outcome is indeed an equilibrium. 
We finally argue that this equilibrium is essentially unique. 
The proof that the canonical outcome is an equilibrium is by a joint
induction on all the properties mentioned. The proof of each of these properties depends intricately on
the others, inductively. This is summarized in \pref{sec:claims}. 

The essence of all this is that any one round looks like a single item auction in the sense that for each agent there is a critical price above which he prefers to lose that item and below which he prefers to win it. (At the critical price she could either have a strict preference for winning over losing or be indifferent.) Monotonicity of prices as the auction progresses is shown by supposing for the sake of contradiction that the price of the first item is less than that of the second\footnote{This is the only case that we have to consider, by induction.} and showing that this leads to a profitable deviation for the loser of the first item. This leads to a contradiction since (inductively) we have shown that the outcome must be an equilibrium. 
The construction of such a profitable deviation itself relies inductively on other properties. 
The sequence of wins, that is that one agent wins all the items in the beginning followed by the other agent winning the rest of the items, 
is also proved by contradiction via a profitable deviation. 
The utility function is discontinuous at the points of transition where an agent can afford to win one extra item. 
We need to show that this discontinuity propagates in a controlled manner, all of which makes the whole proof 
quite technically challenging.


\paragraph{Organization} 
In \pref{sec:prelim} we give formal definitions of the game, the notions of equilibria we consider and analyze the cases with a few items. 
We give the construction of the canonical outcome, the properties on which we do the joint induction to prove that this outcome is an equilibrium
and their dependence on each other in \pref{sec:claims}. In
\pref{sec:sketch} we give proof sketches of some of the important steps in
the induction. \pref{sec:clinching} contains a comparison of the outcome of the sequential auction and that of the adaptive clinching auction. 
\pref{sec:future} discusses extensions and future work. The full proof
is in \pref{app:proof}.

\section{Preliminaries} 
\label{sec:prelim}
\paragraph{Problem definition}
Suppose a central principal wants to auction $k$ identical items to two
budget constrained agents by running $k$ first-price sealed-bid auctions
in a sequential manner. We seek to understand the equilibrium allocation
and prices of this sequential auction. Here, we define an agent's
utility to be the valuation of the items she gets {\em plus her
remaining budget} at the end of the auction. If the budget is exceeded
then the utility is negative infinity. Note that this is not the
standard definition of quasi-linear utilities. Using this definition of
utilities is essential for obtaining the natural monotonicity of
utilities with respect to the budgets, which will play an important role
in our proofs.

Further, we let $B_i$ denote the
budget of agent $i$ and let $v_i$ denote the value of agent $i$ for
winning each item. We will assume that the values are the
dominant factor in the sense that the agents' primary goal is to
maximize the number of items that they get in the sequential auction. 
Minimizing the total price paid for the items is only a secondary
goal, conditioned on getting the same number of items. More precisely, we
will enforce this assumption by assuming $v_1, v_2 > 2^k B_1, 2^k B_2$.
In this case, any change in the prices will be overcompensated by even a
tiny chance of getting an extra item. We consider the complete information case, 
that is, we assume that each agent knows all the valuations and the budgets. 

Formally, this is an instance of an {\em extensive game with complete
information and simultaneous moves}. (See \cite{mas1995microeconomic} or
\cite{osborne1994course} for a formal definition.) 
In each round, both the agents move simultaneously and submit a bid.  
The outcome is a first-price auction: an item is awarded to the agent with the higher bid and the amount of his bid is deducted from his budget. 
For technical reason, we will allow each agent $i$ to bid $b$ and $b+$
for each $0 \le b < B_i$, where bidding $b+$ means bidding
infinitesimally larger than $b$. If one agent bids $b$ and the other
agent bids $b+$ then the agent bidding $b+$ wins the item and is charged
$b$. This treatment is essential for
the existence of Nash equilibrium in the first price auction and
therefore essential for the sequential auction. See
\cite{paes2011sequential}, for example, for other use of this treatment
in the literature. 
Finally if both agents bid the same, then the item is awarded to an agent uniformly at random. 

The natural solution concept for such games is a {\em subgame perfect equilibrium}. 
Given any partial history of an extensive game, the remainder of the game is yet another extensive game, 
called a subgame of the original game. A subgame perfect equilibrium is a Nash equilibrium of the extensive game such that for any partial history of the game, the induced strategies of the Nash equilibrium is a Nash equilibrium of the induced subgame. Intuitively, this requires that given what has already occurred, each player acts rationally with respect to the remainder of the game. This rules out agents playing threats that are not credible. 

\paragraph{Equilibrium refinement}
The sequential auction we consider has unnatural subgame perfect
equilibria. So we consider refinements of the solution concept that rule
out such equilibria. For instance, in a first-price sealed-bid auction
of a single item where the agents' values are $v_1 > v_2$, agent $1$
bidding $b+$ and agent $2$ bidding $b$ is an equilibrium for any $v_1 >
b \ge v_2$. However, we want to rule out the unnatural equilibria where
$b > v_2$ because in these equilibria agent $2$ bids above her value.

A widely-used refinement for finite games is the {\em
trembling-hand-perfection} originally proposed in
\cite{selten1975reexamination}, which is defined as follows:

\begin{definition}[Trembling-Hand-Perfection] \label{def:thp}
    An equilibrium strategy profile $(\sigma, \tau)$ of a two-player
    finite game $G$ is {\em trembling-hand-perfect} if there exists a
    sequence $\{(\sigma_j, \tau_j)\}_j$ of completely mixed strategy
    profiles such that $(\sigma_j, \tau_j)$ converges $(\sigma, \tau)$
    as $j$ goes to infinity, and $\sigma$ is a best reply to every
    $\tau_j$, and $\tau$ is a best reply to every $\sigma_j$ in the
    sequence.
\end{definition}

Here, a completely mixed strategy profile is one in which every strategy
of every player is played with some positive probability. The
positive, and presumably tiny, probability of playing strategies that
are not in the equilibrium profile corresponds to the fact that players
make small mistakes from time to time (trembling hand). So what
trembling-hand-perfection indicates is that the equilibrium strategy
profile is robust to such small mistakes in the sense that using the
equilibrium strategy maximizes the agent's utility even if we assume the
other agent makes small mistakes. 

Note that this is not the standard definition of
trembling-hand-perfection in the literature. We choose this definition
because it is more convenient for our discussion.  See
\cite{selten1975reexamination} for the standard definition of
trembling-hand-perfection and the equivalence of the different
definitions.

Despite the many appealing properties of trembling-hand-perfection in
finite games, there is no well-accepted definition of
trembling-hand-perfection in the literature for games with continuous
strategy space such as the sequential auctions in this paper. Moreover,
if we use the definition of finite games directly, then {\em
trembling-hand-perfect equilibria may not exist}.
Indeed, a trembling-hand-perfect equilibrium may not exists even for the
first-price auction with a single item and two agents. Suppose the
values are $v_1 > v_2$. Then, consider the equilibrium agent $1$ bidding
$v_2 +$ and agent $2$ bidding $v_2$ (it is easy to rule out the other
equilibria). But bidding $v_2$ is dominated by bidding $v_2 - \epsilon$
for agent $2$: if both cases lose, then they yield the same utility; if
bidding $v_2$ wins and bidding $v_2 - \epsilon$ loses, then winning at
$v_2$ is still dominated because it yield utility zero; finally if both
cases win, then bidding $v_2 - \epsilon$ is strictly better because it
pays less. Therefore, this cannot be a trembling-hand-perfect
equilibrium because folklore result asserts the strategy of each agent
in a trembling-hand-perfect equilibrium must be non-dominated.

A commonly used remedy to the absence of a theory for
trembling-hand-perfection in continuous games is to consider a sequence
of discretized version of the continuous game that converges to it, and
then analyze the limit of the trembling-hand-perfect equilibria of the
discretized games (e.g.~\cite{bagnoli1989provision,broecker1990credit}).
This approach, however, is very difficult to apply to
games with complicated structures such as the sequential auction with
arbitrary number of items.

In order to settle this problem, we will use a slightly weaker
equilibrium refinement, which we refer to as the {\em
semi-trembling-hand-perfect equilibrium}.

\begin{definition}[Semi-Trembling-Hand-Perfection]
    \label{def:sthp}
    An equilibrium strategy profile $(\sigma, \tau)$ of a two-player
    game $G$ is {\em semi-trembling-hand-perfect} if there exists a
    sequence $\{(\sigma_j, \tau_j)\}_j$ of completely mixed strategies
    such that $(\sigma_j, \tau_j)$ converges to $(\sigma, \tau)$, and
    the best reply to $\sigma_j$ converges to $\tau$, and the best reply
    to $\tau_j$ converges to $\sigma$ as $j$ goes to infinity.
\end{definition}

As we can see by comparing \pref{def:sthp} and \pref{def:thp}, the
notion of semi-trembling-hand-perfection is still trying to model the
robustness of the equilibrium strategies with respect to  small
mistakes made by the agents, but in a weaker sense: when the other agent
makes small mistakes, we no longer require the best reply to be exactly
the equilibrium strategy; however, the best reply has to  converge to the
equilibrium strategy as the other agent makes fewer and fewer mistakes.

We will construct a subgame-perfect equilibrium that is unique after the
refinement of semi-trembling-hand-perfection.

\subsection{Warm-up: Few-item cases} \label{sec:fewitem}

Let us first examine the cases with only $1$, $2$ or $3$ items in order
to build our intuition for the problem.

\paragraph{$k = 1$} 
First of all, let us consider the simplest case of a single item. In
this case, the problem becomes the classic first price auction.
Therefore, suppose the agent values are $B_1 < B_2$, then the unique
Nash equilibrium that survives the iterated elimination of dominated
strategies is where agent $1$ bids $B_1$ and agent $2$ bids $B_1^+$ and
wins the item. In the case of $B_1 = B_2$, both agent will bid their
budgets and we have a tie.

\paragraph{$k = 2$} 
Next, let us move on to the more interesting case of a two-item
sequential auction. In order to examine how the allocation of items
changes as the ratio between the budgets changes, let us fix $B_2 = 1$
and gradually increase $B_1$ starting from $0$. 

If $B_1 < \frac{1}{2}$, then agent $2$ has enough budget to get both
items.  Since the value of a single item dominates any changes in the
payment, agent $2$ wins both items by bidding $B^+_1$ in both rounds.

If $\frac{1}{2} < B_1 < 1$, then agent $1$ can guarantee herself an item
by bidding $B_1$ in both rounds. So the best strategy for agent $2$ is
to let agent $1$ win the first item and pay the highest possible price.
There are two types of credible threats that agent $2$ could use to set
the price for the first item. The first threat is to offer agent $1$ the
first item at price $B_2 - B_1$, threatening that if agent $2$ wins the
first item at this price then she has sufficient budget remaining to win
the second item as well. The second threat is to offer the first item at
price $\frac{B_1}{2}$, threatening that if agent $1$ do not take the
first item at this price then she would need to pay a higher price in
order to win the second item. 

Agent $2$ uses the larger of these two threats to set the price in round 1. 
If $\frac{1}{2} < B_1 <\frac{2}{3}$, then agent $2$ uses the first threat. 
Agent $1$ wins the first item at price $B_2 - B_1$ and 
agent $2$ wins the second item paying agent $1$'s remaining budget, $2 B_1 - B_2$. 
If $\frac{2}{3} < B_1 < 1$, then agent $2$ exploits the second threat. 
Agent $1$ wins the first item at price $\frac{1}{2} B_1$ and
agent $2$ wins the second item also at price $\frac{1}{2} B_1$.

By symmetry the case of $B_1 > 1$ is identical with the roles of the agents swapped. 
The equilibrium allocations in various cases are summarized in \pref{tab:twoitem}.

\begin{table}
    \centering
    \caption{Equilibrium strategies for $k = 2$}
    \medskip
    \label{tab:twoitem}
	\begin{tabular}{cccc}
		\hline \\ [-2ex]
		~~Phase~~ & Budget ratio ($B_1/B_2$) & Round 1 & Round 2 \\
        [.5ex]
		\hline \\ [-2ex]
		$1$ & $(0, \frac{1}{2})$ & 2 wins at $B_1$ & 2 wins at
        $B_1$ \\ [.5ex] 
        $2$ & $(\frac{1}{2}, \frac{2}{3})$ & 1 wins at $B_2 - B_1$
        & 2 wins at $2 B_1 - B_2$ \\ [.5ex] 
        $3$ & $(\frac{2}{3}, 1)$ & 1 wins at $\frac{1}{2} B_1$ & 2
        wins at $\frac{1}{2} B_1$ \\ [.5ex]
		\hline
    \end{tabular}
\end{table}

\paragraph{$k = 3$} 
The case with three items is much more complicated 
in the sense that there are more possibilities of allocation sequences. 
Here we briefly demonstrate these different
allocation sequences and the intuition behind them. Again, we fix
the budget of agent $2$ to be $B_2 = 1$ and gradually increase agent
$1$'s budget starting from $0$.

The first phase is when $B_1 < \frac{1}{3}$, where agent $2$ has
enough budget to win all three items by bidding $B_1^+$ in all three rounds. 

The second phase is when $\frac{1}{3} < B_1 < \frac{3}{8}$, where
agent $1$ has enough budget to obtain one item.  In this phase, agent
$2$ forces agent $1$ to get the first item at price $B_1 - 2 B_2$
and win the next two items cheaply paying agent $1$'s remaining budget. 
The threat used by agent $2$ is that she could win the remaining two
items as well if agent $1$ does not accept this offer. If agent $2$ gets
the first item at this price, then the induced subgame falls into phase
$1$ of the two-item case.

The third phase is when $\frac{3}{8} < B_1 < \frac{1}{2}$, where
agent $2$ forces agent $1$ to get the first item at price
$\frac{1}{4} B_2$ via a different threat: if agent $2$ wins the first
item at this price, then the induced subgame falls into phase $2$ of 
the two-item case. 
In fact the threshold at which the induced subgame falls into phase $2$
of the two-item case is $B_2 - \frac{3}{2} B_1$.  
However for this to be a credible threat by agent $2$, she must weakly
prefer winning the item at this price to losing it. 
On the one hand, the utility of agent $2$ for winning the first item  at some price $p$, 
assuming that the induced subgame falls into phase $2$ of the two-item case, 
is $2 v_2 + B_2 - p - (2B_1 - (B_2 - p)) = 2v_2 + 2 B_2 - 2 B_1 - 2p$. 
On the other hand, the utility of agent $2$ for losing the first item at price $p$ is
$2 v_2 + B_2 - 2 (B_1 - p) = 2 v_2 + B_2 - 2 B_1 + 2p$. 
The price $p = \frac{1}{4} B_2$ is obtained by equating the two and solving $2v_2 + 2 B_2 - 2 B_1 - 2p = 2 v_2 + B_2 - 2 B_1 + 2p$. 

The rest of the phases are similar in spirit to the above. In the
fourth phase, agent $2$ forces agent $1$ to get the first item at
price $\frac{1}{2} B_1$, threatening that if she wins the item at this
price then the induced subgame falls into phase $3$ of the two-item
case. In the fifth to the seventh phases, it is still the case that
agent $1$ gets one item and agent $2$ gets two. However, agent $1$ now
has enough budget to be in the dominant position in the price
competition. So agent $1$ forces agent $2$ to pay higher prices for
the first two items and then wins the last item, paying agent $2$'s remaining budget. 
Depending on the budget of agent $1$, 
the prices she can set are determined by the phases in the two-item subgame they would end up in, 
if she wins the first item at those prices. 
We do not discuss further details of these threats but summarize the
outcomes in \pref{tab:threeitem}. 

\begin{table}
    \centering
    \caption{Equilibrium strategies for $k = 3$}
    \label{tab:threeitem}
    \medskip
	\begin{tabular}{cccc}
		\hline \\ [-1.5ex]
		Budget ratio ($B_1/B_2$) & Round 1 & Round 2 & Round 3\\ [1ex]
		\hline \\ [-1.5ex]
		$(0, \frac{1}{3})$ & 2 wins at $B_1$ & 2 wins at $B_1$ & 2 wins
        at $B_1$ \\ [1ex]
		$(\frac{1}{3}, \frac{3}{8})$ & 1 wins at $B_2 - 2 B_1$ & 2 wins
        at $3 B_1 - B_2$ & 2 wins at $3 B_1 - B_2$ \\ [1ex]
		$(\frac{3}{8}, \frac{1}{2})$ & 1 wins at $\frac{1}{4} B_2$ & 2
        wins at $B_1 - \frac{1}{4} B_2$ & 2 wins at $B_1 - \frac{1}{4}
        B_2$ \\ [1ex]
        $(\frac{1}{2}, \frac{2}{3})$\footnotemark[1] & 1 wins at
        $\frac{1}{2} B_1$ & 2 wins at $\frac{1}{2} B_1$ & 2 wins at
        $\frac{1}{2} B_1$ \\ [1ex]
		$(\frac{2}{3}, \frac{5}{6})$ & 2 wins at $\frac{1}{3} B_2$ & 2
        wins at $\frac{1}{3} B_2$ & 1 wins at $\frac{1}{3} B_2$ \\
        [1ex]
		$(\frac{5}{6}, \frac{9}{10})$ & 2 wins at $B_1 - \frac{1}{2}
        B_2$ & 2 wins at $\frac{3}{4} B_2 - \frac{1}{2} B_1$ & 1 wins at
        $\frac{3}{4} B_2 - \frac{1}{2} B_1$ \\ [1ex]
		$(\frac{9}{10}, 1)$ & 2 wins at $B_1 - \frac{1}{2} B_2$ & 2 wins
        at $2 B_1 - \frac{3}{2} B_2$ & 1 wins at $3 B_2 - 3 B_1$ \\
        [1ex]
		\hline
    \end{tabular}
    \begin{flushleft}
        {\footnotesize\footnotemark[1] In this case, due to a tie in the
        first round, it could also be that agent $2$ wins the first item
        at price $\frac{1}{2} B_1$, agent $1$ wins the second item at
        price $\frac{1}{2} B_1$, and agent $2$ wins the last item at price
        $\frac{1}{2} B_1$. Since the agents get the same number of items
        at the same prices in both the outcomes, we will break ties in
        some particular way in order to get more consistent structures,
        which are explained in more details in the next section.}
    \end{flushleft}
\end{table}

\paragraph{Observations} We end this section with a few
observations regarding the equilibrium outcomes discussed above. 

\begin{enumerate}
    \item First of all, agent $1$ and agent $2$ get $k_1$ and $k_2$
        items respectively ($k_1 + k_2 = k$)  if and only if the
        budgets satisfy $\frac{k_1}{k_2 + 1} < \frac{B_1}{B_2} <
        \frac{k_1 + 1}{k_2}$.\footnote{We omit the boundary cases
        $\frac{B_1}{B_2} = \frac{k_1}{k_2 + 1}$ here. As we will see
        in the next section, in these boundary cases the agents 
        keep making the same bids until one of them runs out of budget.
        We handle these cases separately.} Intuitively, these conditions
        can be interpreted as follows: if the average price an agent,
        say, agent $1$, can afford for $k_1$ items 
        is strictly greater than the average price agent $2$
        can afford for $k_2 +1 $  items, then agent $1$ can guarantee 
        winning $k_1 $ items. For instance, agent
        $1$ could keep bidding $\frac{B_1}{k_1}$. 
    \item It is not difficult to verify that in each of the
        phases of different allocation sequences, the prices of the
        items are non-increasing in the number of rounds. 
    \item Intriguingly, it is always the case that one of
        the agents wins her share of the items and then the other
        agent wins the rest of the items. In other words, there are no
        interleaving of winners as one might imagine.
    \item Finally, we can easily examine that the utility of an agent
        in the equilibrium is monotonically  increasing in the agent's
        budget and non-increasing in the budget of the other agent. 
        Although this observation seems intuitive, an example in
        \cite{benoit2001multiple} indicates that the utilities might not
        be monotone in the agents' budgets. However, in
        \cite{benoit2001multiple} the utility is defined as the total
        value of the items an agent gets minus the total price that she
        pays.  Indeed, with this definition, an agent's utility could
        counter-intuitively decrease as her budget increases. For
        instance, suppose we fix agent $2$'s budget to be
        $1$ and gradually increases agent $1$'s budget from
        $\frac{2}{3}$ to $1$ in the two-item case. Agent $1$ always gets one item 
        but agent $2$ forces agent $1$ to pay higher and higher price for the item as
        agent $1$'s budget increases. In this paper, we use a different
        notion, an agent's utility is defined to be
        the total value of the items she gets plus the remaining budget
        that she has at the end. With this definition, we can show the desired monotonicity,
        which plays a crucial role in our analysis. 
\end{enumerate}

\section{General Case}
\label{sec:claims}
In this section, we consider the general case of arbitrary number of items and outline the proof structure
and the high level ideas. The easiest part of the whole proof is showing  how the number 
of items won by each agent in any equilibrium depends on the budgets. 
The intuition is that the number of items an agent wins should be approximately proportional to his budget. 
Consider $\frac{k_1}{k - k_1 + 1}$ as a function of $k_1$ and observe that it is monotonically increasing. 
So there must exists a unique $k_1 \in \Z_{\ge 0}$ such that $\frac{k_1}{k
- k_1 + 1} \le \frac{B_1}{B_2} < \frac{k_1 + 1}{k - k_1}$.  If we 
let $k_2 = k - k_1$, then either $\frac{k_1}{k_2 + 1} < \frac{B_1}{B_2}
< \frac{k_1 + 1}{k_2}$, or $\frac{k_1}{k_2 + 1} = \frac{B_1}{B_2}$. We
will handle these two cases separately as follows.

\begin{proposition} \label{prop:alloco}
    Suppose $k_1, k_2 \in \Z_{\ge 0}$ satisfy $k_1 + k_2 = k$ and
    $\frac{k_1}{k_2 + 1} < \frac{B_1}{B_2} < \frac{k_1 + 1}{k_2}$. Then,
    in any subgame perfect equilibrium of the sequential auction, 
    agent $1$ gets $k_1$ items and agent $2$ gets $k_2$ items.
\end{proposition}

\begin{proof} 
    Regardless of the strategy of agent $2$, agent $1$ can guarantee
    $k_1$ items for herself by bidding $\frac{B_1}{k_1}$ until her
    budget is exhausted. This is because at this price, agent 2 can
    afford to win at most $k_2$ items after which her remaining budget
    is strictly less than $\frac{B_1}{k_1}$. Therefore any outcome in
    which agent 1 gets less than $k_1$ items cannot be an equilibrium
    regardless of the price paid, since we assumed that $v_1 > B_1$. By
    symmetry agent 2 gets at least $k_2$ items in any equilibrium. Since
    $k_1 + k_2 = k$, it must be the case that agent $1$ gets exactly
    $k_1$ items and agent $2$ gets exactly $k_2$ items. 
%
%
\end{proof}

\begin{proposition}
    Suppose $k_1, k_2 \in \Z_{\ge 0}$ satisfy $k_1 + k_2 = k$ and
    $\frac{k_1}{k_2 + 1} = \frac{B_1}{B_2}$. Then, in any equilibrium, 
    agent 1 gets at least $k_1- 1$ items and agent 2 gets at least $k_2$ items. 
    (The tie-breaking comes into effect for the extra item.) 
 \end{proposition}

The proof of this proposition is essentially the same as that of
\pref{prop:alloco}. 

\subsection{Canonical Outcome}

We now construct a canonical outcome of the game. The heart of the proof
is showing that this is a subgame perfect equilibrium. This involves
exploiting several other properties of this outcome. We then show that
under the refinement of semi-trembling-hand-perfection this is the only
equilibrium. 

The canonical outcome is defined recursively. We start with the
single-item case and assume w.l.o.g.~that $B_1 \ge B_2$. In this case,
the sequential auction is simply the first-price sealed-bid auction. 

\begin{quote}
    {\em Canonical Outcome (base case):~} The canonical outcome 
    for a single-item is agent $1$ bidding $B_2+$ and agent $2$ bidding
    $B_2$ if $B_1 > B_2$, and both agents bidding $B_1 = B_2$ otherwise.
\end{quote}

Suppose we have defined the canonical outcome for $k-1$ or fewer items.
In order to recursively define the canonical outcome for $k$ items, we
need some definitions. We use $i$ and $-i$ to denote agents. If $i$ is
$1$ then $-i$ is $2$ and vice versa. 

\begin{definition}[Utility]
    Let $U^{(k)}_i(B_i, B_{-i})$ denote the utility (expected utility in
    case tie-breaking comes into effect) of agent $i$ in the canonical
    outcome when the budgets are $B_i$ and $B_{-i}$. 
\end{definition}

\begin{definition}[Winning Utility]
    Let $W^{(k)}_i (B_i, B_{-i}, p) \eqdef v_i + U_i^{(k-1)}(B_i -
    p, B_{-i})$ denote the {\em winning utility} of agent $i$, that is,
    the utility when agent $i$ wins the first item at price $p$ and both
    agents follow the canonical outcome in the remaining sequential
    auction with $k-1$ items.
\end{definition}

\begin{definition}[Losing Utility]
    Let $L^{(k)}_i (B_i, B_{-i}, p) \eqdef U^{(k-1)}(B_i, B_{-i} - p)$
    denote the {\em losing utility} of agent $i$, that is, the utility
    when agent $i$ wins the first item at price $p$ and both agents
    follow the canonical outcome in the remaining sequential auction
    with $k-1$ items.
\end{definition}

We now state a monotonicity property of these functions which is used in
the recursive definition of the canonical outcome. This property is
proved along with other properties later. 

\begin{proposition} \label{prop:utilmon}
    $U^{(k)}_i(B_i, B_{-i})$ is increasing in $B_i$, and non-increasing
    in $B_{-i}$.
\end{proposition}

\begin{proposition} \label{prop:wlmon}
    $W^{(k)}_i(B_i, B_{-i}, p)$ is decreasing in $p$ while
    $L^{(k)}_i(B_i, B_{-i}, p)$ is non-decreasing in $p$, and both are
    increasing in $B_i$ and non-increasing in $B_{-i}$.
\end{proposition}

It is easy to see that \pref{prop:utilmon}  implies \pref{prop:wlmon}.
%
Given the monotonicity in \pref{prop:wlmon}, we define the {\em critical prices}  of the agents for the first round.

\begin{definition}[Critical Prices] \label{def:critical}
    There exists a unique price $\price{k}{i}$
    s.t.~for any $p > \price{k}{i}$, $W^{(k)}_i(B_i, B_{-i}, p)
    < L^{(k)}_i(B_i, B_{-i}, p)$, and for any $p < \price{k}{i}$,
    $W^{(k)}_i(B_i, B_{-i}, p) > L^{(k)}_i(B_i, B_{-i}, p)$. We will
    refer to $\price{k}{i}$ as the {\em critical price} of agent $i$.
    \footnote{We note that it might be the case that even if $p$
    achieves its maximum value $B_i$, we still have $W^{(k)}_i(B_i,
    B_{-i}, p) > L^{(k)}_i(B_i, B_{-i}, p)$. But it is not difficult to
    prove this can only happen when $B_i < \frac{B_{-i}}{k}$, in which
    case it is clear agent $-i$ will win all the items paying $B_i$ per
    item. We will omit this trivial case in our discussion and assume
    $p_i$ always exists.}
\end{definition}

The definition of critical prices is identical to the definition of {\em
critical values} in \cite{pitchik1986budget} for the two-item case.
Here, we extend the definition to arbitrary number of items. 
The critical price $\price{k}{i}$ is similar, although not identical, to
an agent's valuation in the single-item auction in the following sense:
agent $i$ is willing to get the first item at any price lower than
the critical price $\price{k}{i}$, while she has no interest in winning 
the first item at a price higher than $\price{k}{i}$.
Note that one or both of $W^{(k)}_i$ and $L^{(k)}_i$ could be
discontinuous at $\price{k}{i}$. So the monotonicity in
\pref{prop:wlmon} does not have further indication on whether agent $i$
prefers winning or losing the first item at price $\price{k}{i}$. 
In the later sections, we do show that an agent weakly prefers winning to losing at the critical price.  
In order to show this, we will need to utilize a few subtle structures of the canonical outcome. 
 
We are now ready to complete the construction of the canonical outcome.

\begin{quote}
    {\em Canonical Outcome (recursive step):~} Suppose that the
    canonical outcome is defined for $k-1$ items and we have computed
    the critical prices. Assume w.l.o.g that $\price{k}{i} \ge
    \price{k}{-i}$. Then, the canonical outcome of the $k$-item case is
    agent $i$ bidding $\price{k}{-i}+$ and agent $-i$ bidding
    $\price{k}{-i}$ and both agents following the canonical outcome in
    the subgame of $k-1$ items,  if $\price{k}{i} >  \price{k}{-i}$.
    Otherwise, the canonical outcome is both agents bidding
    $\price{k}{i} = \price{k}{-i}$ and following the canonical outcome
    in the resulting subgame of $k-1$ items.
\end{quote}

\subsection{Properties of the Canonical Equilibrium}

The main structural theorem in this paper is the following. 

\begin{theorem} \label{thm:canonicaleq} 
    The canonical outcome is a subgame perfect equilibrium of the
    sequential auction with $k$ items.  
\end{theorem}

The proof of this theorem is by a joint induction on \pref{prop:utilmon}
and several other properties of canonical outcomes. We now detail these
properties and specify how the proof of each of these properties
inductively depends on the others. 

The next two properties characterize which agent wins which round in the
canonical outcome of the sequential auctions and the monotonicity of the
prices paid in the canonical outcome.

\begin{proposition} \label{prop:alloc}
    Suppose $k_1, k_2 \in \Z_{\ge 0}$ satisfy $k_1 + k_2 = k$ and
    $\frac{k_1}{k_2 + 1} \leq  \frac{B_1}{B_2} < \frac{k_1 + 1}{k_2}$.
    Based on the different cases, the following happens in the canonical
    outcome: 
    \begin{description} 
        \item [Case 1 (Type I tie-breaking)] 
            $\frac{k_1}{k_2 + 1} = \frac{B_1}{B_2}$. In this case, both
            agents will keep bidding $p^* \eqdef \frac{B_1}{k_1} =
            \frac{B_2}{k_2 + 1}$ until one of the agents runs out of her
            budget, and then the other agent will win the rest of the
            items for free. 
        \item [Case 2] 
            $\frac{k_1}{k_2 + 1} < \frac{B_1}{B_2} < \frac{k_1 + 1}{k_2
            + 1}$. In this case, agent 1 wins the first $k_1$ items and
            then agent 2 wins the rest. \footnotemark[1]
        \item [Case 3 (Type II-A tie-breaking)]  
            $\frac{B_1}{B_2} =\frac{k_1 + 1}{k_2 + 1}$. 
            In this case, both agents will keep bidding $p^* \eqdef
            \frac{B_1}{k_1 + 1} = \frac{B_2}{k_2 + 1}$ until one of the
            agent's budget becomes $p^*$, and then the other agent will
            get the remaining items by bidding $p^*+$. \footnotemark[2] 
        \item [Case 4] 
            $\frac{k_1 + 1}{k_2 + 1} < \frac{B_1}{B_2} < \frac{k_1 +
            1}{k_2}$. In this case agent $2$ wins the first $k_2$ items
            and then agent $1$ wins the rest. \footnotemark[1]
    \end{description} 

    \noindent{\footnotesize \footnotemark[1] ~ When $\frac{B_1}{B_2}$
    is sufficiently close to $\frac{k_1 + 1}{k_2 + 1}$, we may end up in
    another type tie-breaking case. However, in this tie-breaking case,
    we can assume w.l.o.g.~the allocation is as asserted in
    \pref{prop:alloc} without changing the agents' utilities. We
    will explain in details in \pref{app:proof}. Such treatment will be
    convenient for our proofs.}

    \noindent{\footnotesize \footnotemark[2] ~ In this tie-breaking
    case, we can assume either agent $1$ gets the first $k_1$ items and
    then agent $2$ gets the remaining $k_2$ items, or the other way
    around, without changing the utilities of the agents. Such treatment
    will be convenient for our proofs.}
\end{proposition}


\begin{proposition} \label{prop:pricemon}
     The prices paid in each round of the canonical outcome is
     non-increasing as the auction proceeds.  
 \end{proposition}   
   
Finally, we will also show the following property about the continuity
of the utility functions in the agents' budgets. 

\begin{proposition} \label{prop:utilcont}
    For $i = 1, 2$, $U^{(k)}_i$ is continuous in both $B_i$ and
    $B_{-i}$, except when $\frac{B_i}{B_{-i}} = \frac{k_i}{k_{-i} + 1}$
    for $k_i, k_{-i} \in \Z_{\ge 0}$ such that $k_i + k_{-i} = k$.
    Moreover, 
    $$\lim_{B_{-i} \rightarrow \frac{k_{-i} + 1}{k_i} B_i^-} U^{(k)}_i(B_i,
    B_{-i}) = k_i v_i ~,~ \lim_{B_{-i} \rightarrow \frac{k_{-i} + 1}{k_i}
    B_i^+} U^{(k)}_i(B_i, B_{-i}) = (k_i - 1) v_i + B_i$$
    $$\lim_{B_i \rightarrow \frac{k_i}{k_{-i} + 1} B_{-i}^-} U^{(k)}_i(B_i,
    B_{-i}) = (k_i - 1) v_i + \frac{k_i}{k_{-i} + 1} B_{-i} ~,~ \lim_{B_i
    \rightarrow \frac{k_i}{k_{-i} + 1} B_{-i}^-} U^{(k)}_i(B_i, B_{-i}) =
    k_i v_i$$
\end{proposition}

If $\frac{B_i}{B_{-i}} = \frac{k_i}{k_{-i} + 1}$, then the tie-breaking
comes into effect and the utilities are random variables. We can
calculate the expected utilities, and these are summarized later.

\paragraph{Dependencies}

For the sake of presentation, we will let $\peq{k}$ denote
\pref{thm:canonicaleq} for $k$ items. Let $\pallocseq{k}$ denote
\pref{prop:alloc} for $k$ items. Let $\ppricemon{k}$ denote
\pref{prop:pricemon} for $k$ items. Let $\putilmon{k}$ and
$\putilcont{k}$ denote \pref{prop:utilmon} and \pref{prop:utilcont}
respectively. The joint induction proceeds by assuming each of the
properties $\peq {k-1}$ to $\putilcont {k-1} $ and proving properties
$\peq {k}$ to $\putilcont {k} $. The dependencies are summarized below:
\begin{itemize}
    \item[] $\peq{k}$  depends on $\peq{k-1}$, $\putilmon{k-1}$, and $\putilcont{k-1}$.
    \item[] $\pallocseq{k}$ depends on $\peq{k}$, $\putilmon{k-1}$, and $\putilcont{k-1}$. 
    \item[] $\ppricemon{k}$ depends on $\peq{k}$, $\pallocseq{\le k}$, $\putilmon{k-1}$, and $\putilcont{k-1}$;
    \item[] $\putilmon{k}$ and $\putilcont{k}$ depend on $\putilmon{k-1}$, $\putilcont{k-1}$, and $\pallocseq{k}$.
\end{itemize}

\subsection{Uniqueness of equilibrium} 
Note that when $\price{k}{i} > \price{k}{-i}$ (or symmetrically the
other way around), we face the problem of equilibrium section. We can
argue that agent $i$ bidding $\price{k}{-i}+$
and agent $-i$ bidding $\price{k}{-i}$ is the only ``stable''
equilibrium in the sense that it is the unique
semi-trembling-hand-perfect equilibrium. 

\begin{proposition} \label{prop:unique}
    The canonical equilibrium is the unique semi-trembling-hand-perfect
    and subgame-perfect equilibrium.
\end{proposition}

\begin{figure}
\centerline{\includegraphics[width=.4\textwidth]{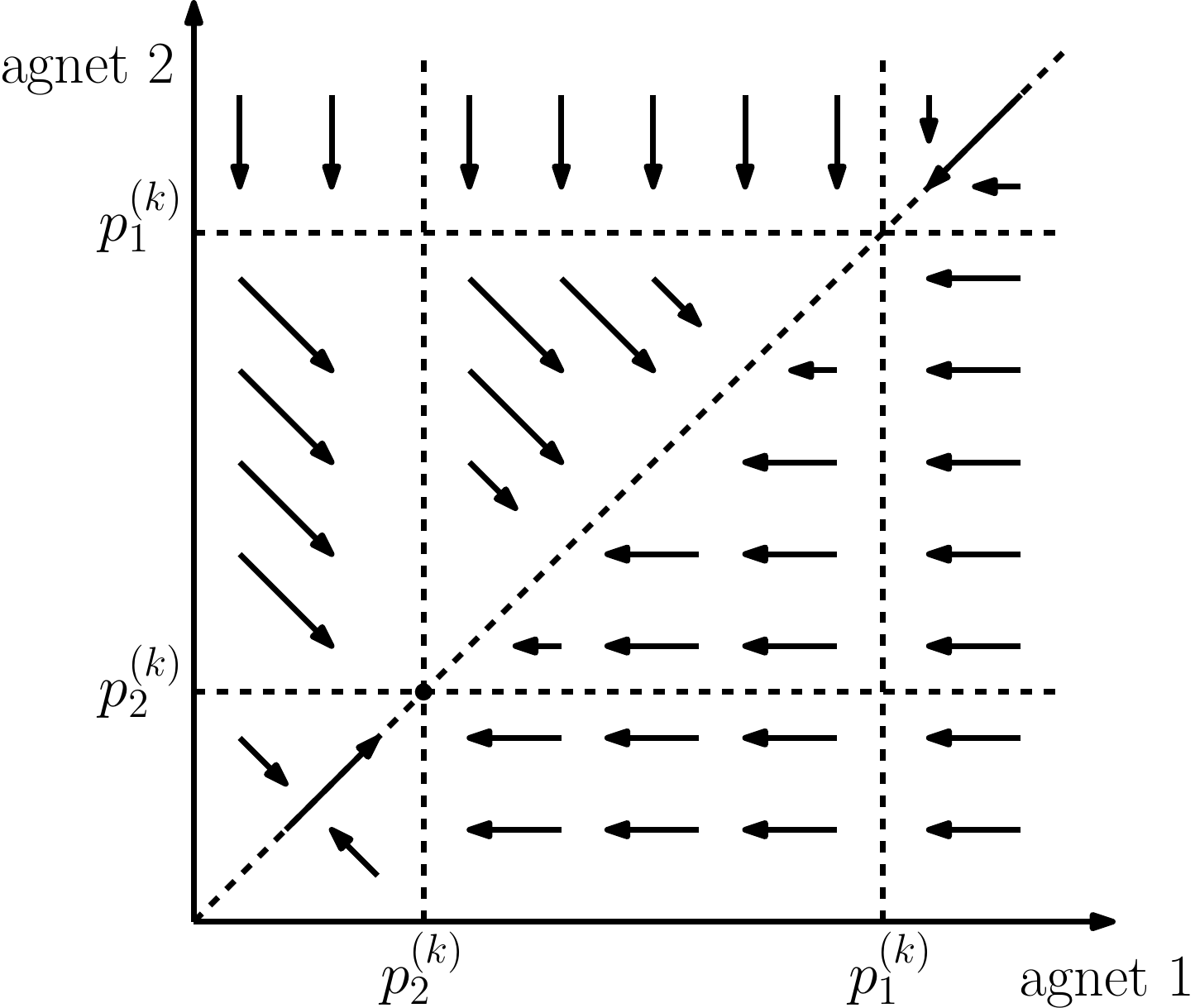}}
\caption{Proof sketch of uniqueness of the canonical equilibrium by
analyzing the dynamic of the game.}
\label{fig:unique}
\end{figure}

The analysis is very similar to the equilibrium refinement in a first
price auction and therefore we will defer the proof to \pref{app:unique}.
Here we will give an informal explanation by presenting
the dynamic of the game in \pref{fig:unique}. Suppose w.l.o.g.~that
$\price{k}{1} > \price{k}{2}$, and agent $1$ bids $b_1$ and agent $2$
bids $b_2$. If $b_1 > b_2$ and $b_1 > \price{k}{2}$, then agent $1$
wants to decease her bid while agent $2$ has no incentive of changing
her strategy. If $b_1 < b_2$ and $b_2 > \price{k}{1}$, then agent $1$
has no incentive to change her bid while agent $2$ wants to decrease her
bid to lose the item. If $b_1 < b_2 < \price{k}{1}$, then agent $1$
wants to increase her bid to wins the item while agent $2$ wants to
decrease her bid to lower the price. Finally, if $b_2 < b_1 <
\price{k}{2}$, then agent $1$ wants to decrease her bid to lower the
price, while agent $2$ wants to increase her bid to win the item.
Summarizing these cases, it is easy to see the only ``stable'' point is
agent $1$ bidding $\price{k}{2}+$ and agent $2$ bidding $\price{k}{2}$.

\subsection{First-price vs.~second-price} \label{sec:secondprice}

We can also consider the second-price version of the sequential auction,
where the winner and the price are chosen using the second-price rule in
each round: the agent with the higher bid wins and pays the other
agent's bid.

In the second-price version, the canonical outcome we construct in
this paper is still a subgame-perfect equilibrium, via almost identical
arguments. However, it is no longer the stable one because underbidding
is weakly-dominated under the second-price rule. In the second-price
auction of a single item, truthful bidding is a stable equilibrium as
it is a dominant-strategy equilibrium. We can use the
critical-price methodology in this paper to analyze the case of multiple
items. Roughly speaking, when it comes to $2$ or more items, the agent
with the lower critical price may have incentive to overbid (unlike
the single-item case) in order to deplete the other agent's budget.
In such case, the first-round bids in the equilibria must be the winner
bidding $b+$ while the loser bidding $b$ for some $b$ between the
critical prices. Further, since underbidding is weakly-dominated in the
second-price rule, in the stable equilibrium $b$ must equal the higher
one of the critical prices. In other words, it is conceivable that the
stable equilibrium is both agent bidding the higher one of the agents'
critical prices in each round, rather than the lower one as in the
first-price version. This observation is implicitly proved for
the value-dominant case of two items in \cite{pitchik1986budget}. Also,
see \cite{pitchik1988perfect} for experimental results comparing the
first-price and second-price sequential auctions. We shall not discuss
any further the second-price version but leave it as an interesting open
question to study the equilibrium of the second-price sequential
auction.

\section{Proof Sketches}\label{sec:sketch}
In this section, we will sketch the proofs of some of the important
implications outlined in \pref{sec:claims}. The complete joint induction
of these claims is in \pref{app:proof}. 

\subsection{Two-Phase Winner Sequence}

Consider the two-phase winner sequence described in
\pref{prop:alloc}. The simplest case is the type I tie-breaking, which
we restate and prove as follows.

\begin{proposition}[Type I tie-breaking] \label{prop:alloctie}
    Suppose $k_1, k_2 \in \Z_{\ge 0}$ satisfy $k_1 + k_2 = k$ and
    $\frac{k_1}{k_2 + 1} = \frac{B_1}{B_2}$. Then, in the canonical
    outcome, both agents keep bidding $p^* \eqdef
    \frac{B_1}{k_1} = \frac{B_2}{k_2 + 1}$ until one agent is
    out of budget, and then the other agent wins the rest for free.
\end{proposition}

\begin{proof}
    Intuitively, this is the case where the budget ratio lies at the
    intersection of the region where agent $1$ gets $k_1$ items and
    agent $2$ gets $k_2$ items and the region where agent $1$ gets $k_1
    - 1$ items and agent $2$ gets $k_2 + 1$ items. In other words, agent
    $1$ can guarantee $k_1 - 1$ items and agent $2$ can guarantee $k_2$
    items for sure, and both agents are competing for the extra item by
    keep bidding the highest rational bid they have.

    Now we formally prove the claim. We first prove that
    $\price{k}{1} = \price{k}{2} = p^*$. Consider any price $p^* +
    \epsilon$ for sufficiently small $\epsilon > 0$.  It is easy to
    verify that $\frac{k_1 - 2}{k_2 + 2} < \frac{B_1 - p -
    \epsilon}{B_2} < \frac{k_1 - 1}{k_2 + 1}$ and $\frac{k_1}{k_2} <
    \frac{B_1}{B_2 - p - \epsilon} < \frac{k_1 + 1}{k_2 - 1}$ when
    $\epsilon$ is sufficiently small.
    So by \pref{prop:alloc}, if agent $1$ wins the first item at price
    $p^* + \epsilon$, then in the subgame she only gets $k_1 - 2$
    items and thus $k_1 - 1$ items in total; on the other hand, if agent
    $1$ loses the first item at price $p^* + \epsilon$, then in the
    subgame she gets $k_1$ items. Therefore, we have $W^{(k)}_1(B_1,
    B_2, p^* + \epsilon) < L^{(k)}_1(B_1, B_2, p^* + \epsilon)$ since
    $v_1 \gg B_1$. So by
    the definition of $\price{k}{1}$ we have $\price{k}{1} \le p^*$.
    Similarly, we can show $\price{k}{1} \ge p^*$. So we have
    $\price{k}{1} = p^*$. Via an almost identical proof we can show that
    $\price{k}{2} = p^*$. 
    
    Now note that $\frac{B_1 - p^*}{B_2} =
    \frac{k_1 - 1}{k_2 + 1}$ and $\frac{B_1}{B_2 - p^*} =
    \frac{k_1}{k_2}$. So we can recursively apply the same argument in
    the subgames to finish the proof of \pref{prop:alloctie}.
\end{proof}

Next, we sketch the proof of the other three cases. We first
clarify the tie-breaking case described in the footnote of
\pref{prop:alloc}.

\begin{proposition} \label{prop:allocseq}
    Suppose $k_1, k_2 \in \Z_{\ge 0}$ satisfy $k_1 + k_2 =
    k$ and $\frac{k_i}{k_{-i} + 1} < \frac{B_i}{B_{-i}} \le \frac{k_i +
    1}{k_{-i} + 1}$. Then, in the canonical outcome it is always the
    case that agent $i$ wins the first  $k_i$ items, and then agent $-i$
    gets the remaining $k_{-i}$ items paying agent $i$'s remaining
    budget for each item. Here, there are two tie-breaking caveats:
    \begin{description}
        \item[\em Type II-A Tie-breaking] If $\frac{B_1}{B_2} =
            \frac{k_1 + 1}{k_2 + 1}$, then in the canonical
            outcome both agents keep bidding $p^* \eqdef
            \frac{B_1}{k_1 + 1} = \frac{B_2}{k_2 + 1}$ until one
            agent's budget becomes $p^*$, and then the other
            agent gets the remaining item by bidding $p^*+$. 

            In this tie-breaking case, we can assume either agent $1$
            gets the first $k_1$ items and then agent $2$ gets the
            remaining $k_2$ items, or the other way around, without
            changing the utilities of the agents.

        \item[\em Type II-B Tie-breaking] If $\frac{B_i}{B_{-i}}$ is
            smaller than and sufficiently close to $\frac{k_i +
            1}{k_{-i} + 1}$, then in the canonical outcome both
            agents keep bidding $p^* \eqdef \frac{B_i}{k_i + 1}$
            until either agent $i$ gets $k_i$ items or agent $-i$ gets
            $k_{-i} -1$ items. In the former case, agent $-i$ then wins
            the remaining items paying agent $i$'s remaining budget
            $p^*$. In the latter case, agent $i$ then keeps bidding
            $p^*+$ until she gets $k_i$ items, and then agent $−i$ wins
            the remaining items paying agent $i$'s remaining budget
            $p^*$. 
            
            In this tie-breaking case, we can assume agent $i$ gets
            $k_i$ items and then agent $−i$ gets the remaining $k_{-i}$
            items, without changing the utilities of the agents.
    \end{description}
\end{proposition}

\begin{proof}[Sketch of \pref{prop:allocseq}]
    Here we only show that one of the agent gets her share of
    the items first and then the other agent gets the rest. Once we
    have shown this, it is easy to deduce who gets the items first.

    Suppose for contradiction that agent $1$ gets the first item
    paying $\price{k}{2}$,
    and then agent $2$ gets the next $k_2$ items, and finally agent
    $1$ gets the rest. Briefly speaking, we first show that in
    this case, the first item must be sold at a price that is lower than
    the average price that agent $1$ pays. Then, we argue that since the
    first item is so cheap, agent $1$ must have a profitable deviation
    by winning the first item. 

    There will be several cases we need to consider in the complete
    proof. Here we describe one of the cases. We assume the
    following: (1) no tie in the first round: $\price{k}{1} >
    \price{k}{2}$; (2) agent $2$ is indifferent between winning and
    losing the first item at $\price{k}{2}$; (3) if agent $2$ wins the
    first item paying $\price{k}{2}$, then in the subgame agent $2$ wins
    the next $k_2 - 1$ items and then agent $2$ wins the rest.

    By the second assumption, agent $2$ has the same remaining
    budget regardless of whether she wins the first item or not. If
    agent $2$ wins the first item, agent $1$ gets $k_1$ items at
    price equal to agent $2$'s remaining budget. If agent $1$ wins the
    first item, then agent $1$ gets $k_1 - 1$ items at a price equal to
    agent $2$'s remaining budget, and one item at price $\price{k}{2}$.
    Further, agent $1$ must strictly prefer winning the first item.
    Therefore, $\price{k}{2}$ must be strictly smaller than agent $2$'s
    remaining budget. Now by the price monotonicity
    (\pref{prop:pricemon}) of $k-1$ items, $\price{k}{2}$ is strictly
    smaller than any other price in the assumed equilibrium. 
    Now we have a contradiction by constructing the following
    equilibrium. Agent $2$ could win the first item paying a price
    greater than $\price{k}{2}$ but lower than the prices of the other
    items in the assumed equilibrium. Then, agent 2 could keep bidding
    $\price{k}{2}$ until agent $1$ gets at least one item, say item
    $j^*$. Finally, agent $2$ could follow the canonical outcome
    thereafter. In the first $j^*$ rounds of this deviation, agent $1$
    gets one item at price $\price{k}{2}$, and agent $2$ gets $j^* - 1$
    items paying roughly $\price{k}{2}$ per item. In the first $j^*$
    rounds of the assumed equilibrium, however, agent $1$ also gets one
    item at price $\price{k}{2}$, and agent $2$ gets $j^* - 1$ items
    paying strictly greater than $\price{k}{2}$ per item. Therefore,
    agent $2$ is strictly better off in the deviation by the
    monotonicity of an agent's utility in her budget.
\end{proof}

\subsection{Monotonicity of Utilities}

The complete proof of \pref{prop:utilmon} requires a more detailed case
analysis. Here we sketch the proof by describing the analysis for one of
the cases.

For the sake of the discussion here, we assume  the following: (1) agent
$1$ wins the first item at agent $2$'s critical price in the canonical
outcome; (2) agent $2$ is indifferent between winning and losing the
first item at price $\price{k}{2}$; (3) the above two assumptions still
hold when agent $2$'s budget increases from $B_2$ to $B_2 + \epsilon$ or
when agent $1$'s budget increases from $B_1$ to $B_1 + \epsilon$ for
sufficiently small $\epsilon > 0$; (4) we have inductively proved \pref{prop:utilmon} for the case of $k-1$ items.

First, consider increasing agent $2$'s budget 
from $B_2$ to $B_2 + \epsilon$. Note that agent $2$'s utility when
her budget is $B_2$ is $U^{(k)}_2(B_2, B_1) = L^{(k)}_2(B_2,
B_1, \price{k}{2}) = W^{(k)}_2(B_2, B_1, \price{k}{2})$ as she is
indifferent between winning and losing at price $\price{k}{2}$.
Moreover, we have $L^{(k)}_2(B_2 + \epsilon, B_1, \price{k}{2}) >
L^{(k)}_2(B_2, B_1, \price{k}{2}) = U^{(k)}_2(B_2, B_1)$ where we get
the inequality by comparing the budgets after the first round when agent
$2$'s budget is $B_2$ and $B_2 + \epsilon$. Similarly, we have
$W^{(k)}_2(B_2 + \epsilon, B_1, \price{k}{2}) > W^{(k)}_2(B_2, B_1,
\price{k}{2}) = U^{(k)}_2(B_2, B_1)$. So by bidding $\price{k}{2}$ agent
$2$ could guarantee strictly more utility than $U^{(k)}_2(B_2, B_1)$
when her budget is $B_2 + \epsilon$. Thus, in the canonical outcome 
agent $2$'s utility will strictly increase when her budget increases.

Next, consider increasing agent $1$'s budget 
from $B_1$ to $B_1 + \epsilon$. We will let $p'(\epsilon)$
denote the critical price of agent $2$ when agent $1$'s budget is
$B_1 + \epsilon$ instead of $B_1$. We will show that
$p'(\epsilon) < \price{k}{2} + \epsilon$. Note that $L^{(k)}_2(B_2,
B_1 + \epsilon, \price{k}{2} + \epsilon) = L^{(k)}_2(B_2, B_1,
\price{k}{2})$ because we will end up in the same subgame after the
first round in both cases. On the other hand, we have
$W^{(k)}_2(B_2, B_1 + \epsilon, \price{k}{2} + \epsilon) <
W^{(k)}_2(B_2, B_1, \price{k}{2})$ by comparing the budgets after the
first round in the two cases. So by our assumption that agent $2$ is
indifferent at price $\price{k}{2}$ when the budgets are $B_1$ and
$B_2$, we have $L^{(k)}_2(B_2, B_1 + \epsilon, \price{k}{2} + \epsilon)
> W^{(k)}_2(B_2, B_1 + \epsilon, \price{k}{2} + \epsilon)$. Since agent
$2$ is indifferent at price $p'(\epsilon)$ when agent $1$'s budget is
$B_1 + \epsilon$, we have $p'(\epsilon) < \price{k}{2} + \epsilon$.
Hence, we get that
$$U^{(k)}_2(B_2, B_1 + \epsilon) = U^{(k-1)}_2(B_2, B_1 + \epsilon -
p'(\epsilon)) \le U^{(k-1)}_2(B_2, B_1 - \price{p}{2}) =
U^{(k)}_2(B_2, B_1) \enspace.$$

Therefore, we have proved the desired monotonicity for this case.
The analysis of the other cases are similar in spirit to the above
arguments. But some special treatments are needed on a case-by-case
basis. The full proof is in \pref{app:proof}.

\section{Comparison with the clinching auction}\label{sec:clinching} 
In this section we compare the outcome of the sequential auction with
that of the clinching auction of \cite{ausubel2004efficient}, as modified by
\cite{dobzinski2011multi} to accommodate budgets.  At a high level, the
clinching auction proceeds as follows.  It starts with a per-unit price of
$0$, and slowly raises this until some agent becomes critical for
keeping overall demand above supply, i.e. the total demand from {\em
  other} agents drops below the supply.  When this happens, we
allocate to this agent at the current price, until the supply is
reduced back to below other agents' total demand (this process is the
``clinching'' from the auction's name).  We then return to slowly
raising the per-unit price.  This process continues until the supply
has been exhausted.

When analyzing sequential auctions, we focus on the case with just $2$
bidders, both of whom have values exceeding their budgets.  In this
case, the clinching auction simplifies to the following process: items
are allocated one-by-one, and at a given round, if there are $k$ items
left, and the remaining budgets of agents $1$ and $2$ are $B_1$ and
$B_2$, respectively, an item is allocated to the agent with the higher
budget at a price of $\min(B_1,B_2)/k$.

We begin by considering cases with small $k$, as we did with
sequential auctions, and then present some observations on the
outcomes of the clinching auction for general $k$.  For $k=1$, we can
see from the description above that if $B_1<B_2$ (without loss of
generality), then agent $2$ will simply win the single item at a price
of $B_1$.  Consider the case where $k=2$, again assuming that
$B_1<B_2$.  Then in the first round, agent $2$ wins the item at a
price of $B_2/2$; the outcome of the second round simply depends on
whether or not paying this price depletes agent $2$'s budget below
$B_1$, i.e. whether or not $B_2 > 3B_1/2$.  Note that this induction
is much simpler than the one for sequential auctions -- once we know
the budgets of the agents, we may immediately infer the winner and
payment in the current round, which tells us the budgets for the next
round.  As such, we omit further details, and simply present the
outcomes and prices for $k=2$ and $k=3$ in
\pref{tab:clinchingtwoitem} and \pref{tab:clinchingthreeitem},
respectively.

We finish this section with some observations about the outcomes of
the clinching auction for general $k$.

\begin{table}
    \centering
    \caption{Clinching auction outcomes for $k = 2$}
    \label{tab:clinchingtwoitem}
    \medskip
	\begin{tabular}{cccc}
		\hline \\ [-2ex]
		Budget ratio ($B_1/B_2$) & Round 1 & Round 2 \\
        [.5ex]
		\hline \\ [-2ex]
		 $(0, \frac{2}{3})$ & 2 wins at $\frac{B_1}{2}$ & 2 wins at
        $B_1$ \\ [.5ex] 
         $(\frac{2}{3}, 1)$ & 2 wins at $\frac{B_1}{2}$
        & 1 wins at $1-\frac{B_1}{2}$ \\ [.5ex]
		\hline
    \end{tabular}
\end{table}

\begin{table}
    \centering
    \caption{Clinching auction outcomes for $k = 3$}
    \label{tab:clinchingthreeitem}
    \medskip
	\begin{tabular}{cccc}
		\hline \\ [-1.5ex]
		Budget ratio ($B_1/B_2$) & Round 1 & Round 2 & Round 3\\ [1ex]
		\hline \\ [-1.5ex]
		$(0, \frac{6}{11})$ & 2 wins at $\frac{B_1}{3}$ & 2 wins at $\frac{B_1}{2}$ & 2 wins
        at $B_1$ \\ [1ex]
		$(\frac{6}{11}, \frac{3}{4})$ & 2 wins at $\frac{B_1}{3}$ & 2 wins
        at $\frac{B_1}{2}$ & 1 wins at $1-\frac{5B_1}{6}$ \\ [1ex]
		$(\frac{3}{4},1)$ & 2 wins at $\frac{B_1}{3}$ & 1
        wins at $\frac{3-B_1}{6}$ & 2 wins at $\frac{5B_1-3}{6}$ \\
        [1ex]
		\hline
    \end{tabular}
\end{table}

\begin{enumerate}
\item 
    First, we observe that the order of allocation is quite
  different from that in the sequential auction.  While in the
  sequential auction winners are never interleaved, in the clinching
  auction winners are nearly always interleaved.  In particular, it is
  easy to verify that once there are two consecutive rounds where the
  winning agent changes, the winner will change in every round
  thereafter. Thus, we find that in general the clinching auction
  repeatedly allocates to the agent with the higher budget until the
  two budgets are ``close,'' and then proceeds to alternate between
  the two agents.
\item 
Unlike sequential auctions, where the prices are non-increasing
  as we proceed to later rounds, the very definition of the clinching
  auction ensures that prices with only increase as we proceed to
  later rounds.
\item 
The fact that the clinching auction generally alternates between
  the two agents when allocating items might give the impression that
  it is more ``fair'' than the sequential auction; considering how the
  items are divided between the agents, however, gives a much
  different impression.  In particular, the clinching auction
  increases the number of items it allocates to agent $1$ as $B_1$
  increases much more quickly than the sequential auction does.  For
  example, agent $1$ cannot win all the items in a sequential auction
  until $B_1 > k\cdot B_2$, but in the clinching auction does so with
  only $B_1 >H_k\cdot B_2 \approx \log k \cdot B_2$.  More generally,
  if agent $1$ controls a $p$-fraction of the total wealth
  (i.e. $p=B_1/(B_1+B_2)$), the fraction of items agent $1$ can expect
  to win is $k_1/k\approx p$ under the sequential auction, but
  $k_1/k\approx 1-\exp(\frac{2p-1}{p-1})/2$ in the clinching auction.
  See \pref{fig:sequentialVsClinching} for a graphical
  comparison of these two quantities.
\end{enumerate}

\begin{figure}
    \centering
    \includegraphics[width=.5\textwidth]{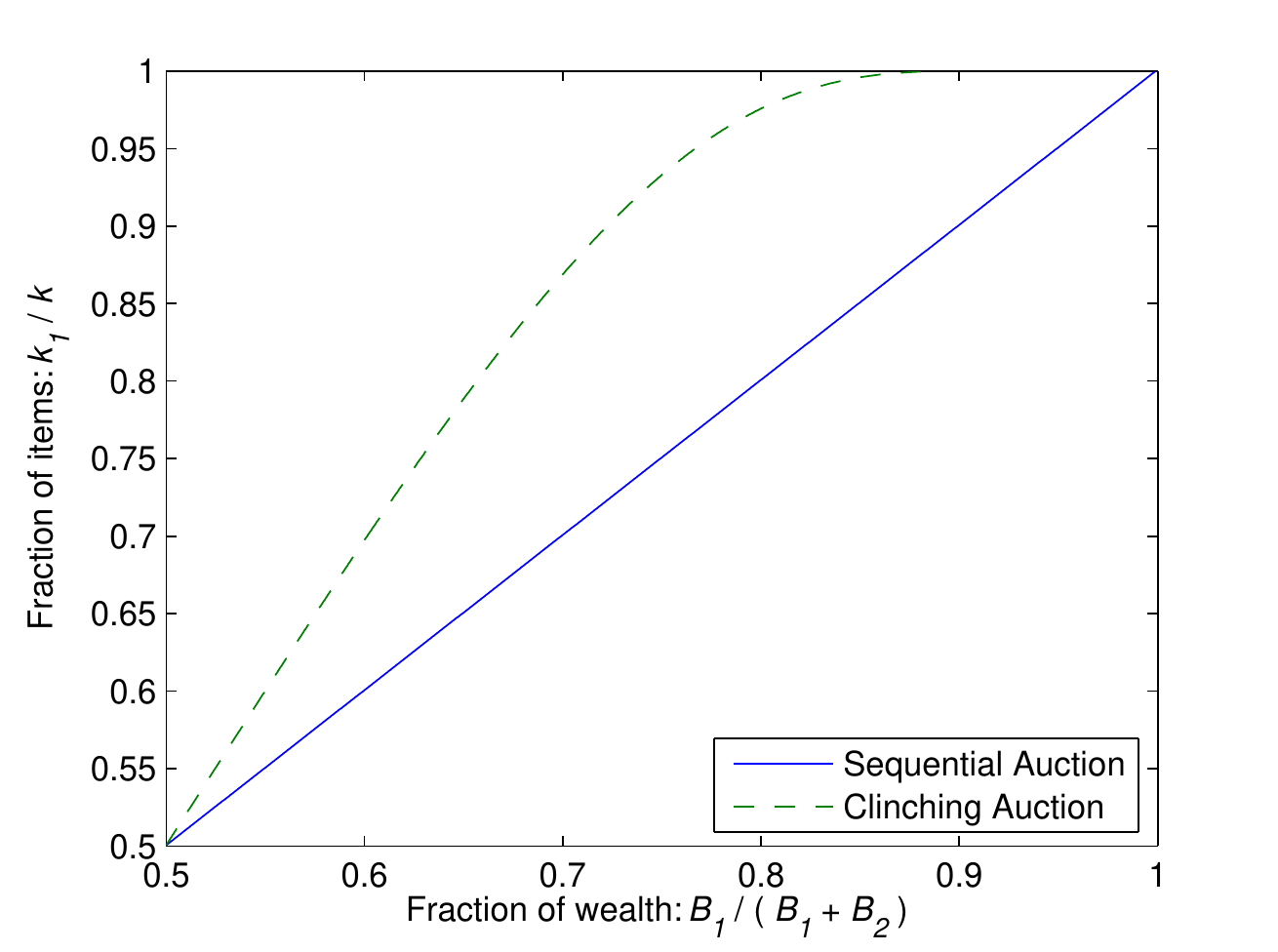}
    \caption{A comparison of how the sequential and clinching auctions
    compare in terms of the fraction of items an agent can expect to
    receive as a function of the fraction of the total wealth they
    control under each.}
    \label{fig:sequentialVsClinching}
\end{figure}

\section{Conclusion and future directions} 

%
\label{sec:future}

As mentioned in the introduction, we consider the most basic setting of sequential auctions with budget constrained bidders  in order to analyze, for the first time, 
the case of arbitrary number of items. 
Now that we understand this basic case, one can hope to extend it to more settings. 
In particular as ``next steps'',  the following are interesting questions that our work raises. 

\begin{itemize} 
\item
Even for the case of two bidders, we assume that the values are much bigger than the budgets, in essence to ensure that the primary 
goal of the agents is to win as many items as possible and the price paid is only a secondary goal. Can we extend our results to arbitrary valuations? 
\item
A crucial property in our results was the sequence of wins. Even with 3 bidders, it is not clear how this generalizes. 
How to extend our results to an arbitrary number of bidders is of course the big question. 
\item
Another direction in which we can relax our model is to consider incomplete information. This would become more interesting when the valuations are comparable to the budget. 
\end{itemize} 

Finally, we feel we are still far from a thorough understanding of even the two agent case presented here. 
It would be insightful to further simplify the proof we currently have.

\nocite{pitchik1988perfect}

\bibliographystyle{plainnat}
\bibliography{seqauctionbiblio}

\newpage

\appendix

\section{Joint Induction}
\label{app:proof}

\subsection{Notations}

As we will see later in this section, the probability trial of tossing
(at most $m + n + 1$) coins until either we get $m+1$ head or we get
$n+1$ tails will play an important role in the tie-breaking cases of the
equilibrium. Therefore,
we will define the following notation to denote the probability of
getting $n+1$ tails for presentation convenience.

\begin{definition}
    For any $m, n \in \Z_{\ge 0}$, we let $\phi(m, n)$ denote the
    probability that we get at least $n + 1$ heads if we repeatedly toss
    a fair coin $m + n + 1$ times.
\end{definition}

We will prove two straightforward properties of the probability $\phi(m,
n)$.

\begin{lemma} \label{lem:1}
    For any $m, n \in \Z_{\ge 0}$, 
    $$\phi(m, n) = \sum_{j=0}^m {n + j \choose n}
    \left(\frac{1}{2}\right)^{n + j + 1} \enspace.$$
\end{lemma}

\begin{proof}
    First, the number of different head/tail sequences such that we get
    $n$ heads and $j$ tails in the first $n + j$ tosses and get another
    head in the last toss is ${n + j \choose n}$. Each of these
    sequences happens with probability $\left(\frac{1}{2}\right)^{n + j
    + 1}$. So the probability that we get $n + 1$ heads and before that
    get exactly $j$ tails is ${n + j \choose n}
    \left(\frac{1}{2}\right)^{n + j + 1}$. Summing up for $j$ from $0$
    to $m$ proves the lemma.
\end{proof}

\begin{lemma} \label{lem:2}
    For any $m, n \in \Z$ such that $m \ge 1$ and $n \ge 0$ we have
    $$\phi(m, n) - \phi(m - 1, n + 1) = {n + m + 1 \choose n + 1}
    \left(\frac{1}{2}\right)^{m + n + 1} \enspace.$$
\end{lemma}

\begin{proof}
    By the definition of $\phi$, we have that $\phi(m, n) - \phi(m - 1,
    n + 1)$ equals the probability that in $m + n + 1$ tosses we get
    exactly $n + 1$ heads. Note that there are ${m + n + 1 \choose n +
    1}$ such head/tail sequences and each happens with probability
    $\left(\frac{1}{2}\right)^{m + n + 1}$. So we have proved the lemma.
%
%
\end{proof}
    
\subsection{Main Proof}

Now let us present the joint inductive proof of 
\pref{thm:canonicaleq}, \pref{prop:allocseq} (the elaborated version of
\pref{prop:alloc}, Case 2 - 4) ,
\pref{prop:pricemon}, \pref{prop:utilmon}, and \pref{prop:utilcont}.
First, let us restate and clarify these claims as follows:

\begin{description}
    \item[$\peq{k}$ (\pref{thm:canonicaleq})] 
        The canonical outcome is a subgame perfect equilibrium of the
        sequential auction with $k$ items.  

        \medskip

    \item[$\pallocseq{k}$ (\pref{prop:allocseq})]
        Suppose $k_1, k_2 \in \Z_{\ge 0}$ satisfy $k_1 + k_2 = k$ and
        $\frac{k_i}{k_{-i} + 1} < \frac{B_i}{B_{-i}} \le \frac{k_i +
        1}{k_{-i} + 1}$. Then, in the canonical outcome it is always the
        case that agent $i$ wins the first  $k_i$ items, and then agent
        $-i$ gets the remaining $k_{-i}$ items paying agent $i$'s
        remaining budget for each item. Here, there are two tie-breaking
        caveats:
        \begin{description}
            \item[\em Type II-A Tie-breaking] If $\frac{B_1}{B_2} =
                \frac{k_1 + 1}{k_2 + 1}$, then in the canonical
                outcome both agents keep bidding $p^* \eqdef
                \frac{B_1}{k_1 + 1} = \frac{B_2}{k_2 + 1}$ until one
                agent's budget becomes $p^*$, and then the other agent
                gets the remaining item by bidding $p^*+$.  
                
                \medskip 
                
                In this tie-breaking case, we can assume either agent
                $1$ gets the first $k_1$ items and then agent $2$ gets
                the remaining $k_2$ items, or the other way around,
                without changing the utilities of the agents.  
                
                \medskip

            \item[\em Type II-B Tie-breaking] If $\frac{B_i}{B_{-i}}$ is
                smaller than and sufficiently close to $\frac{k_i +
                1}{k_{-i} + 1}$, then in the canonical outcome both
                agents keep bidding $p^* \eqdef \frac{B_i}{k_i + 1}$
                until either agent $i$ gets $k_i$ items or agent $-i$
                gets $k_{-i} -1$ items. In the former case, agent $-i$
                then wins the remaining items paying agent $i$'s
                remaining budget $p^*$. In the latter case, agent $i$
                then keeps bidding $p^*+$ until she gets $k_i$ items,
                and then agent $−i$ wins the remaining items paying
                agent $i$'s remaining budget $p^*$. 
            
                \medskip

                In this tie-breaking case, we can assume agent $i$ gets
                $k_i$ items and then agent $−i$ gets the remaining
                $k_{-i}$ items, without changing the utilities of the
                agents.
        \end{description}

        \medskip

    \item[$\ppricemon{k}$ (\pref{prop:pricemon})]
         The prices paid in each round of the canonical outcome is
         non-increasing as the auction proceeds.  
 
        \medskip

    \item[$\putilmon{k}$ (\pref{prop:utilmon})]
        $U^{(k)}_i(B_i, B_{-i})$ is increasing in $B_i$, and
        non-increasing in $B_{-i}$.
 
        \medskip

    \item[$\putilcont{k}$ (\pref{prop:utilcont})]
        $U^{(k)}_i$ is continuous in both $B_i$ and
        $B_{-i}$, except when $\frac{B_i}{B_{-i}} = \frac{k_i}{k_{-i} +
        1}$ for $k_i, k_{-i} \in \Z_{\ge 0}$ such that $k_i + k_{-i} =
        k$.  Moreover, 
        \begin{eqnarray*}
            \lim_{B_{-i} \rightarrow \frac{k_{-i} + 1}{k_i} B_i^-}
            U^{(k)}_i(B_i, B_{-i}) & = & k_i v_i \enspace, \\
            \lim_{B_{-i} \rightarrow \frac{k_{-i} + 1}{k_i} B_i^+}
            U^{(k)}_i(B_i, B_{-i}) & = & (k_i - 1) v_i + B_i
            \enspace, \\ 
            \lim_{B_i \rightarrow \frac{k_i}{k_{-i} + 1} B_{-i}^-}
            U^{(k)}_i(B_i, B_{-i}) & = & (k_i - 1) v_i +
            \frac{k_i}{k_{-i} + 1} B_{-i} \enspace, \\ 
            \lim_{B_i \rightarrow \frac{k_i}{k_{-i} + 1} B_{-i}^-}
            U^{(k)}_i(B_i, B_{-i}) & = & k_i v_i \enspace.
        \end{eqnarray*}

        If $\frac{B_i}{B_{-i}} = \frac{k_i}{k_{-i} + 1}$, then
        $$U^{(k)}_i(B_i, B_{-i}) = (k_i - 1) v_i + \phi(k_{-i}, k_i - 1)
        v_i + \phi(k_i - 1, k_{-i}) B_i - \phi(k_i - 2, k_{-i} + 1)
        B_{-i} \enspace.$$
\end{description}

\medskip

Also, recall that the dependencies can be summarized as follows:

\begin{enumerate} 
    \item[] $\peq{k}$  depends on $\peq{k-1}$, $\putilmon{k-1}$, and
        $\putilcont{k-1}$.
    \item[] $\pallocseq{k}$ depends on $\peq{k}$, $\putilmon{k-1}$, and
        $\putilcont{k-1}$.
    \item[] $\ppricemon{k}$ depends on $\peq{k}$, $\pallocseq{\le k}$,
        $\putilmon{k-1}$, and $\putilcont{k-1}$.
    \item[] $\putilmon{k}$ and $\putilcont{k}$ depend on $\putilmon{k-1}$,
        $\putilcont{k-1}$, and $\pallocseq{k}$.
\end{enumerate}

\paragraph{Base case} 

These propositions in the case of $1$ or $2$ items are easy to verify
from the discussions in \pref{sec:fewitem}. So we will omit the tedious
calculations here.

\paragraph{Inductive step}

Let us move on to the inductive step. Suppose we have prove
the propositions for the cases of $k-1$ or less
items. Now let us consider the case of $k$ items. The rest of the
section is organized as follows. In \pref{sec:wlmon}, we will verify the
monotonicity and continuity of the winning and losing utilities
(\pref{prop:wlmon}) from the monotonicity and continuity of the utility
functions of $k-1$ items. In \pref{sec:criticalprice}, we will prove the
agents weakly prefer winning at the critical prices, which will be
useful for the rest of the proof. In \pref{sec:unique} we will prove
that the canonical outcome is indeed a subgame-perfect equilibrium
($\peq{k-1}, \putilmon{k-1}, \putilcont{k-1} \Rightarrow \peq{k}$). 
In \pref{sec:twophase}, we will prove the two-phase winner sequence in
the canonical outcome ($\pallocseq{k-1}, \putilmon{k-1},
\putilcont{k-1}, \peq{k} \Rightarrow \pallocseq{k}$). 
In \pref{sec:pricemon}, we will explain why the
prices paid for the items weakly declines as the action proceeds
($\ppricemon{k-1}, \putilmon{k-1}, \putilcont{k-1}, \peq{k},
\pallocseq{\le k} \Rightarrow \ppricemon{k}$). Finally, in
\pref{sec:util}, we will analyze the monotonicity and continuity of the
utility function in the sequential auction with $k$ items
($\putilmon{k-1}, \putilcont{k-1}, \pallocseq{k} \Rightarrow
\putilmon{k}, \putilcont{k}$).

\subsubsection{Monotonicity and Continuity of Winning and Losing Utility} 
\label{sec:wlmon}

In this step, we will establish the monotonicity and
continuity of the winning and losing utilities of the $k$-item
sequential auction. 

\begin{lemma}[\pref{prop:wlmon} restated] \label{lem:wlmon}
    $W^{(k)}_i(B_i, B_{-i}, p)$ is decreasing in $p$ while
    $L^{(k)}_i(B_i, B_{-i}, p)$ is non-decreasing in $p$, and both are
    increasing in $B_i$ and non-increasing in $B_{-i}$.
\end{lemma}

\begin{lemma} \label{lem:wlcont}
    $W^{(k)}_i$ and $L^{(k)}_{-i}$ are continuous in $p$
    except when $p = B_i - \frac{k_i}{k_{-i}} B_{-i}$ for $k_1, k_2
    \in \Z_{\ge 0}$ such that $k_1 + k_2 = k$. Moreover, we have
    \begin{eqnarray*} 
        \lim_{p \rightarrow (B_{-i} - \frac{k_{-i}}{k_i}
        B_i)^-} L^{(k)}_i(B_i, B_{-i}, p) & = & (k_i - 1) v_i
        + B_i \enspace, \\
        \lim_{p \rightarrow (B_{-i} - \frac{k_{-i}}{k_i}
        B_i)^+} L^{(k)}_i(B_i, B_{-i}, p) & = & k_i v_i
        \enspace, \\ 
        \lim_{p \rightarrow (B_i - \frac{k_i}{k_{-i}}
        B_{-i})^-} W^{(k)}_i(B_i, B_{-i}, p) & = & (k_i + 1)
        v_i \enspace, \\
        \lim_{p \rightarrow (B_i - \frac{k_i}{k_{-i}}
        B_{-i})^+} W^{(k)}_i(B_i, B_{-i}, p) & = & k_i v_i +
        \frac{k_i}{k_{-i}} B_{-i} \enspace.
    \end{eqnarray*}
\end{lemma}

\begin{proof}[of \pref{lem:wlmon} and \pref{lem:wlcont}]
    By definition, we have $L^{(k)}_i(B_i, B_{-i}, p) = U^{(k-1)}_i(B_i,
    B_{-i} - p)$. Therefore, by \pref{prop:utilmon} and
    \pref{prop:utilcont} of the $k-1$ item case we get the asserted
    monotonicity of $L^{(k)}_i$. Also, we get that 
    $L^{(k)}_i$ is continuous in $p$ except when $p = B_{-i} - 
    \frac{k_i}{k_{-i}} B_i$ for $k_1, k_2 \in \Z_{\ge 0}$ s.t.~$k_1
    + k_2 = k$, and 
    \begin{eqnarray*}
        \lim_{p \rightarrow (B_{-i} - \frac{k_{-i}}{k_i} B_i)^-}
        L^{(k)}_i(B_i, B_{-i}, p) & = & \lim_{B \rightarrow
        \frac{k_{-i}}{k_i} B_i^+} U^{(k-1)}_i(B_i, B) \,=\, (k_i - 1)
        v_i + B_i \enspace, \\
        \lim_{p \rightarrow (B_{-i} - \frac{k_{-i}}{k_i} B_i)^+} 
        L^{(k)}_i(B_i, B_{-i}, p) & = & \lim_{B \rightarrow
        \frac{k_{-i}}{k_i} B_i^-} U^{(k-1)}_i(B_i, B) \,=\, k_i v_i
        \enspace.
    \end{eqnarray*}

    Similarly, $W^{(k)}_i(B_i, B_{-i}, p) = v_i + U^{(k-1)}_i(B_i - p,
    B_{-i})$. So by \pref{prop:utilmon} and \pref{prop:utilcont} we have
    the desired monotonicity and its continuity in $p$ except when
    $B_{-i} = \frac{k_{-i}}{k_i} (B_i - p)$ for $k_1, k_2 \in \Z_{\ge
    0}$ s.t.~$k_1 + k_2 = k$, and
    \begin{eqnarray*}
        \lim_{p \rightarrow (B_i - \frac{k_i}{k_{-i}} B_{-i})^-}
        W^{(k)}_i(B_i, B_{-i}, p) & = & v_i + \lim_{B \rightarrow
        \frac{k_i}{k_{-i}} B_{-i}^+} U^{(k - 1)}_i(B, B_{-i}) \,=\,
        (k_i + 1) v_i \enspace, \\
        \lim_{p \rightarrow (B_i - \frac{k_i}{k_{-i}} B_{-i})^+} 
        W^{(k)}_i(B_i, B_{-i}, p) & = & v_i + \lim_{B \rightarrow
        \frac{k_i}{k_{-i}} B_{-i}^-} U^{(k-1)}_i(B, B_{-i}) \\
        & = & k_i v_i + \frac{k_{-i}}{k_i} B_{-i} \enspace.
    \end{eqnarray*}

    So we have proved the lemmas.
\end{proof}

\subsubsection{Agents Weakly Prefer wining at Critical Prices} 
\label{sec:criticalprice}

In this section, we will prove that the agents always weakly prefer
winning the first item at the critical prices. 
First we need to prove a few lemmas.

\begin{lemma} \label{lem:7}
    For $i = 1, 2$, suppose $p > 0$ satisfies $p = \frac{B_i}{k_i + 1} =
    \frac{B_{-i}}{k_{-i}}$ for $k_1, k_2 \in \Z_{\ge 0}$ such that
    $k_1 + k_2 = k$, then $W^{(k)}_i(B_i, B_{-i}, p) \ge k_i v_i + 
    \phi(k_{-i} - 1, k_i - 1) v_i$.
\end{lemma}

\begin{proof}
    If agent $i$ wins the first item at price $p$, then in the subgame
    of $k-1$ items, the remaining budgets satisfy $\frac{B_i -
    p}{B_{-i}} = \frac{k_i}{k_{-i}}$. So we are in the tie-breaking case
    of \pref{prop:alloctie}. In this case, agent $i$ will get $k_i$
    items for sure, and have probability $\phi(k_{-i} - 1, k_i - 1)$ of
    getting an extra item. Further, the remaining budget is at least
    zero. So we have $W^{(k)}_i(B_i, B_{-i}, p) \ge k_i
    v_i + \phi(k_{-i} - 1, k_i - 1) v_i$.
\end{proof}

\begin{lemma} \label{lem:8}
    Suppose $p > 0$ satisfies that $p = \frac{1}{k_i + 1} B_i =
    \frac{1}{k_{-i}} B_{-i}$ for $k_1, k_2 \in \Z_{\ge 0}$
    s.t.~$k_1 + k_2 = k$, then $k_i v_i + \phi(k_{-i} - 1, k_i) +
    L^{(k)}_i(B_i, B_{-i}, p) \le k_i v_i + B_i + \phi(k_{-i} - 1, k_i)
    v_i$.
\end{lemma}

\begin{proof}
    If agent $i$ loses the first item at price $p$, then in the subgame
    of $k-1$ items, the remaining budges satisfies $\frac{B_i}{B_{-i}} =
    \frac{k_i + 1}{k_{-i} - 1}$. So we are in the tie-breaking case of
    \pref{prop:alloctie}. In this case, agent $i$ will get $k_i$ items
    for sure, and have probability $\phi(k_{-i} - 2, k_i)$ of getting an
    extra item. Further, the remaining budget is at most $B_i$ and at
    least zero. So we have the desired inequality.
\end{proof}

\begin{lemma} \label{lem:6}
    Suppose $p > 0$ satisfies that $p = \frac{1}{k_i + 1} B_i =
    \frac{1}{k_{-i}} B_{-i}$ for $k_1, k_2 \in \Z_{\ge 0}$
    s.t.~$k_1 + k_2 = k$, then $W^{(k)}_i(B_i, B_{-i}, p) >
    L^{(k)}_i(B_i, B_{-i}, p)$.
\end{lemma}

\begin{proof}
    Let us assume $i = 1$ for the sake of presentation.
    In this case, we have
    \begin{align*}
        W^{(k)}_1(B_1, B_2, p) & = v_1 + U^{(k-1)}_1(B_1 - p, B_2) \\
        & = k_1 v_1 + \phi(k_2 - 1, k_1 - 1) v_1 + \phi(k_1 - 1, k_2 -
        1) B_1 - \phi(k_1 - 2, k_2) B_2 \enspace.
    \end{align*}
    and
    \begin{align*}
        L^{(k)}_1(B_1, B_2, p) & = U^{(k-1)}_1(B_1, B_2 - p) \\
        & = k_1 v_1 + \phi(k_2 - 2, k_1) v_1 + \phi(k_1, k_2 - 2) B_1 -
        \phi(k_1 - 1, k_2 - 1) B_2 \enspace.
    \end{align*}

    Further, by \pref{lem:2} we have $\phi(k_2 - 1, k_1 - 1) - \phi(k_2
    - 2, k_1) = {k - 1 \choose k_1} \left(\frac{1}{2}\right)^{k-1}$ and
    $\phi(k_1, k_2 - 2) - \phi(k_1 - 1, k_2 - 1) = {k - 1 \choose k_2 -
    1} \left(\frac{1}{2}\right)^{k-1} = {k - 1 \choose k_1}
    \left(\frac{1}{2}\right)^{k-1}$. So we have
    \begin{eqnarray*}
        W^{(k)}(B_i, B_{-i}, p) - L^{(k)}(B_i, B_{-i}, p) & = &
        (\phi(k_2 - 1, k_1 - 1) - \phi(k_2 - 2, k_1)) v_1 \\
        & & - (\phi(k_1, k_2 - 2) - \phi(k_1 - 1, k_2 - 1)) B_1 \\
        & & + (\phi(k_1 - 1, k_1 - 1) - \phi(k_2 - 2, k_1)) B_2 \\
        & = & {k - 1 \choose k_1} \left(\frac{1}{2}\right)^{k - 1} (v_1
        - B_1 + B_2) ~ > ~ 0 \enspace.
    \end{eqnarray*}

    So we have proved \pref{lem:6}.
\end{proof}

\begin{lemma} \label{lem:price}
    Either both $W^{(k)}_i$ and $L^{(k)}_i$ are continuous at
    $\price{k}{i}$, in which case agent $i$ is indifferent between
    winning or losing the first item at $\price{k}{i}$, or $W^{(k)}_i$
    is discussions at $\price{k}{i}$, in which case agent $i$ strictly
    prefers winning the first item at price $\price{k}{1}$.
\end{lemma}

\begin{proof}
    For the sake of presentation, we will assume w.l.o.g.~that $i = 1$. 
    There are four cases:
    
    \paragraph{Case 1} Both $W^{(k)}_1$ and $L^{(k)}_1$ are continuous at
    $\price{k}{1}$. In this case, we can deduce by the definition of
    $\price{k}{1}$ that $W^{(k)}_1(B_1, B_2, \price{k}{1}) =
    L^{(k)}_1(B_1, B_2, \price{k}{1})$, that is, agent $1$ is
    indifferent between winning and losing the first item at price
    $\price{k}{1}$.

    \paragraph{Case 2} $L^{(k)}_1$ is continuous at $\price{k}{1}$ but
    $W^{(k)}_1$ is not. In this case, by \pref{lem:wlcont} we conclude
    that $\price{k}{1} = B_1 - \frac{k_1}{k_2} B_2$ for $k_1, k_2
    \in \Z_{\ge 0}$ s.t.~$k_1 + k_2 = k$. By our choice of
    $\price{k}{1}$, we have $W^{(k)}_1(B_1, B_2, p) <
    L^{(k)}_1(B_1, B_2, p)$ for any $p > \price{k}{1}$, and
    $W^{(k)}_1(B_1, B_2, p) > L^{(k)}_1(B_1, B_2, p)$ for any $p <
    \price{k}{1}$. So we have 
    $$L^{(k)}_1(B_1, B_2, \price{k}{1}) = \lim_{p \rightarrow
    \price{k}{1}+} L^{(k)}_1(B_1, B_2, p) \ge \lim_{p \rightarrow
    \price{k}{1}+} W^{(k)}_1(B_1, B_2, p) = k_1 v_1 + \frac{k_1}{k_2}
    B_2$$ 
    and 
    $$L^{(k)}_1(B_1, B_2, \price{k}{1}) = \lim_{p \rightarrow
    \price{k}{1}-} L^{(k)}_1(B_1, B_2, p) \le \lim_{p \rightarrow
    \price{k}{1}-} W^{(k)}_i(B_1, B_2, p) = (k_1 + 1) v_1
    \enspace.$$

    By the value range of $L^{(k)}_1$ in each continuous interval in
    \pref{lem:wlcont}, we know that $B_2 - \frac{k_2}{k_1} B_1 <
    \price{k}{1} < B_2 + \frac{k_2 - 1}{k_1 + 1} B_1$ and $k_1 v_1 \le
    L^{(k)}_1(B_1, B_2, \price{k}{1}) \le k_1 v_1 + B_1$. Finally, the
    utility for winning the first item is
    \begin{align*}
        W^{(k)}_1\left(B_1, B_2, \price{k}{1}\right) & = v_1 +
        U^{(k-1)}\left(\frac{k_1}{k_2} B_2, B_2 \right) 
        \ge k_1 v_1 + \phi(k_2 - 1, k_1 - 1) v_1 \\
        & \ge k_1 v_1 + \frac{1}{2^k} v_1 > k_1 v_1 + B_1 \enspace.
    \end{align*}

    Here, the first inequality holds because agent $1$ gets $k_1$ items
    for sure and gets an extra item with probability $\phi(k_2 - 1, k_1
    - 1)$ in this case.
    So agent $1$ strictly prefers winning the first item.

    \paragraph{Case 3}
    $W^{(k)}_1$ is continuous at $\price{k}{1}$ but $L^{(k)}_1$ is not.
    We will argue this case is impossible. Suppose for contradiction
    that this is the case, then $\price{k}{1} = B_2 - \frac{k_2}{k_1}$.
    Similar to the analysis in the second case, we have that.
    $$W^{(k)}_1(B_1, B_2, \price{k}{1}) = \lim_{p \rightarrow
    \price{k}{1}+} W^{(k)}_1(B_1, B_2, p) \le \lim_{p \rightarrow
    \price{k}{1}+} L^{(k)}_1(B_1, B_2, p) = k_1 v_1$$ 
    and 
    $$W^{(k)}_1(B_1, B_2, \price{k}{1}) = \lim_{p \rightarrow
    \price{k}{1}-} W^{(k)}_1(B_1, B_2, p) \ge \lim_{p \rightarrow
    \price{k}{1}-} L^{(k)}_1(B_1, B_2, p) = (k_1 - 1) v_1 + B_1
    \enspace.$$
    However, by \pref{lem:wlcont} we have $B_1 - \frac{k_1 - 2}{k_2 + 2}
    B_2 < \price{k}{1} < B_1 - \frac{k_1 - 1}{k_2 + 1} B_2$.
    When $\price{k}{1}$ is in this range, the value of $W^{(k)}_1$
    is upper bounded by $W^{(k)}_1(B_1, B_2, \price{k}{1}) < (k_1 - 1)
    v_1 + \frac{k_1 - 1}{k_2 + 1} B_2 < (k_1 - 1) v_1 + B_1$.
    So we have a contradiction.

    \paragraph{Case 4} Both $W^{(k)}_1$ and $L^{(k)}_1$ are
    discontinuous at $\price{k}{1}$. Then, by \pref{lem:wlcont} it must
    be the case that $\price{k}{1} = B_2 - \frac{k_2 - 1}{k_1 + 1} B_1 =
    B_1 - \frac{k_1}{k_2} B_2$ for $k_1, k_2 \in \Z_{\ge 0}$
    s.t.~$k_1 + k_2 = k$. So we have that $\price{k}{1} = \frac{1}{k_1 +
    1} B_1 = \frac{1}{k_2} B_2$. By \pref{lem:6}, agent $1$ strictly
    prefers winning the first item at price $\price{k}{1}$ in this
    case.

    \bigskip

    Summarizing the four cases, we have proved the lemma.
\end{proof}

\subsubsection{Subgame-Perfection of the Canonical Outcome}
\label{sec:unique}

In this step, we aim to establish the fact that the canonical
outcome is a subgame-perfect equilibrium. Let us first prove several
lemmas.

\begin{lemma} \label{lem:3}
    $L^{(k)}_i$ is continuous at $\price{k}{i}$.
\end{lemma}

\begin{proof}
    Let us assume w.l.o.g.~that $i = 1$.
    Suppose for contradiction that $L^{(k)}_1$ is discontinuous at
    $\price{k}{1}$. Then, by \pref{lem:price}, $W^{(k)}_1$ must be
    discontinuous at $\price{k}{1}$ as well. By \pref{lem:wlcont} it
    must be the case that $\price{k}{1} = B_2 - \frac{k_2 - 1}{k_1 + 1}
    B_1 = B_1 - \frac{k_1}{k_2} B_2$ for $k_1, k_2 \in \Z_{\ge 0}$
    s.t.~$k_1 + k_2 = k$. So we have that $\price{k}{1} = \frac{1}{k_1 +
    1} B_1 = \frac{1}{k_2} B_2$. So it is easy to verify $W^{(k)}_2$
    and $L^{(k)}_2$ are discontinuous at $\price{k}{1}$ as well, and 
    $\price{k}{i}$ is also the critical price of agent $2$, which
    contradicts our assumption.
\end{proof}

\begin{lemma} \label{lem:5}
    Suppose $\price{k}{1} = \price{k}{2}$. Then, either all of
    $W^{(k)}_1$, $L^{(k)}_1$, $W^{(k)}_2$, $L^{(k)}_2$ are
    continuous at $\price{k}{1} = \price{k}{2}$, or all of them are
    discontinuous at $\price{k}{1} = \price{k}{2}$.
\end{lemma}

\begin{proof}
    Suppose at least one of $W^{(k)}_1$, $L^{(k)}_1$, $W^{(k)}_2$,
    $L^{(k)}_2$ is discontinuous at $\price{k}{1} = \price{k}{2}$. Note
    that by $\pwlmon{k}$, $W^{(k)}_i$ is discontinuous at a price $p$
    if and only if $L^{(k)}_{-i}$ is discontinuous at $p$. So we can
    assume without loss of generality that $W^{(k)}_1$ and $L^{(k)}_2$
    are discontinuous at $\price{k}{1} = \price{k}{2}$. Now by
    $\pprice{k}$, $L^{(k)}_2$ is discontinuous at $\price{k}{1} =
    \price{k}{2}$ implies $W^{(k)}_2$ is discontinuous at $\price{k}{1}
    = \price{k}{2}$ as well, which further indicates $L^{(k)}_1$ is
    discontinuous at $\price{k}{1} = \price{k}{2}$. So we have proved
    the lemma.
\end{proof}

\begin{lemma} \label{lem:4a}
    Suppose $\price{k}{i} > \price{k}{-i}$. Then, agent
    $i$ bidding $\price{k}{-i}+$ and agent $-i$ bidding $\price{k}{-i}$
    and then both agents following the canonical outcome in the
    subgames of $k-1$ items is a subgame perfect equilibrium.
\end{lemma}

\begin{proof}
    It suffices to show that neither of the agents has profitable
    deviations. Let us assume w.l.o.g.~that $i = 1$ for the sake of
    presentation.

    First, we will consider the possible deviation of agent $1$. We note
    that agent $1$ could not benefit from bidding over $\price{k}{2}$
    because that would only increases her price for getting the first
    item. Further, agent $1$ could not benefit from bidding below
    $\price{k}{2}$ because $\price{k}{2} < \price{k}{1}$ implies that
    agent $1$ strictly prefers winning the first item at prices
    $\price{k}{2}$. 
    
    Next, let us consider the possible deviation of agent $2$. Note that
    underbidding has no effect for that agent $1$ would still win the
    first item at $\price{k}{2}$. Further, by \pref{lem:3},
    $L^{(k+1)}_2$ is continuous at $\price{k}{2}$. So we get that
    overbidding has utility at most 
    $$\lim_{p \rightarrow \price{k}{2}+} W^{(k)}_2(B_2, B_1, p) \le
    \lim_{p \rightarrow \price{k}{2}+} L^{(k)}_2(B_2, B_1, p) =
    L^{(k)}_2(B_2, B_1, \price{k}{2}) \enspace.$$

    Therefore, overbidding could not be a profitable deviation for agent
    $2$ either. 
\end{proof}

\begin{lemma} \label{lem:4b}
    If $\price{k}{1} = \price{k}{2}$, then both agents bidding
    $\price{k}{1} = \price{k}{2}$ and then following the canonical
    outcomes in the subgames of $k-1$ items is a subgame-perfect
    equilibrium.
\end{lemma}
    
\begin{proof}
    We will let $p^* \eqdef \price{k}{1} = \price{k}{2}$ for
    convenience. By symmetry, it suffices to prove agent $1$ has no
    profitable deviation. By \pref{lem:5} we only need to consider the
    following two cases.

    The first case is when all of $W^{(k)}_1$, $L^{(k)}_1$, $W^{(k)}_2$,
    $L^{(k)}_2$ are continuous at $p^*$. In this case, both agents are
    indifferent between winning and losing at price $p^*$. Therefore,
    underbidding yield utility $L^{(k)}_1(B_1, B_2, p^*)$, which is
    the same as the utility of bidding $p^*$. Overbidding is strictly
    worse because for any $p > p^*$, we have $W^{(k)}_1(B_1, B_2, p) <
    W^{(k)}_1(B_1, B_2, p^*)$, which again equals the utility of bidding
    $p^*$.
    
    The second case is when all of $W^{(k)}_1$, $L^{(k)}_1$,
    $W^{(k)}_2$, $L^{(k)}_2$ are discontinuous at $p^*$. In this case,
    we have $p^* = B_1 - \frac{k_1}{k_2} B_2 = B_2 - \frac{k_2 - 1}{k_1
    + 1} B_1$ and thus $p^* = \frac{1}{k_1 + 1} B_1 = \frac{1}{k_2}
    B_2$. By \pref{lem:6} agent $1$ strictly prefers winning the first
    item at price $p^*$ than losing it. 
    So agent $1$ will not underbid. If agent $1$ overbids $p > p^*$,
    then her utility is
    \begin{align*}
        W^{(k)}_1(B_1, B_2, p) & ~ = ~ v_1 + U^{(k)}_1(B_1 - p, B_2) \\
        & ~ < ~ v_1 + \lim_{B \rightarrow (\frac{k_1}{k_2} B_2)^-}
        U^{(k)}_1(B, B_2) & & \text{($B_1 - p < \frac{k_1}{k_2} B_2$)} \\
        & ~ = ~ v_1 + (k_1 - 1) v_1 + \frac{k_1}{k_2} B_2 & & 
        \text{(\pref{lem:wlcont})} \\
        & ~ < ~ k_1 v_1 + B_1 \enspace. & & \text{($\frac{B_1}{B_2} \le
        \frac{k_1 + 1}{k_2}$)}
    \end{align*}

    On the other hand, the utility of bidding $p^*$ is 
    \begin{align*}
        U^{(k)}_1(B_1, B_2) & ~ = ~ \frac{1}{2}\left(W^{(k)}_1(B_1, B_2,
        p^*) + L^{(k)}_1(B_1, B_2, \price{k}{1})\right) \\
        & ~ \ge ~ k_1 v_1 + \frac{1}{2} B_1 + \frac{1}{2} \left(
        \phi(k_2 - 1, k_1 - 1) + \phi(k_2 - 1, k_1) \right) v_1 & &
        \text{(\pref{lem:7}, \pref{lem:8})} \\
        & ~ \ge ~ k_1 v_1 + \frac{1}{2} B_1 + \frac{1}{2^k} v_1 \\
        & ~ > ~ k_1 v_1 + B_1 \enspace.
    \end{align*}

    So bidding $p^*$ is strictly better.
\end{proof}

By \pref{lem:4a} and \pref{lem:4b}, we have shown that the canonical
outcome is indeed a subgame-perfect equilibrium of the $k$-item
sequential auction. 

\subsubsection{Two-Phase Winner Sequence} \label{sec:twophase}

In this section, we will present the inductive proof for the two-phase
winner sequence structure in the canonical outcome. Again, we will start
with a few technical lemmas.

\begin{proposition} \label{prop:defaultprice}
    Suppose we are not in the type I tie-breaking case of
    \pref{prop:alloctie}. In other words, there are $k_1, k_2 \in
    \Z_{\ge 0}$ such that $k_1 + k_2 = k$ and $\frac{k_1}{k_2 + 1} <
    \frac{B_1}{B_2} \le \frac{k_1 + 1}{k_2}$. Then, for $i = 1, 2$,
    agent $i$ gets $k_i$ items with average price at most
    $\frac{B_{-i}}{k_{-i} + 1}$ in the canonical outcome.
\end{proposition}

\begin{proof} 
    By symmetry, it suffices to prove it for $i = 1$. Let us consider
    the following strategy for agent $1$: keep bidding $p^* \eqdef
    \frac{B_2}{k_2 + 1}$ until agent $2$'s remaining budget becomes
    $p^*$, and then keep bidding $p^*+$. It is easy to verify this
    strategy guarantees winning at least $k_1$ items regardless of agent
    $2$'s strategy and paying $p^*$ per item. So in the canonical
    outcome, agent $1$ pays at most $p^*$ per item on
    average.
\end{proof}

\begin{lemma} \label{lem:10}
    Suppose the budgets do not fall into the tie-breaking cases of
    \pref{prop:alloctie}, then the budgets after the first round in the
    canonical outcome do not fall into the tie-breaking cases
    either.
\end{lemma}

\begin{proof}
    We may assume w.l.o.g.~that there exists $k_1, k_2 \in \Z_{\ge 0}$
    such that $k_1 + k_2 = k$, and 
    $\frac{B_1}{k_1 + 1} < \frac{B_2}{k_2 + 1} < \frac{B_1}{k_1}$. 

    Note that $\frac{B_1}{k_1 + 1} < \frac{B_2}{k_2 + 1}$ implies
    $\frac{B_1}{k_1 + 1} < \frac{B_2}{k_2}$. So by \pref{prop:alloc} 
    agent $1$ gets $k_1$ items and agent $2$ gets $k_2$ items in the
    canonical outcome. Assume for contradiction that the equilibrium of
    the subgame after the first round do fall into the tie-breaking case
    of \pref{prop:alloctie}. Then, by \pref{prop:alloctie} the expected
    number of items agent $1$ gets will be a non-integral number for
    that in the type I tie-breaking case, there exists some $k'_1, k'_2
    \in \Z_{\ge 0}$, $k'_1 + k'_2 = 1$, such that agent $1$ gets at
    least $k'_1$, and agent $2$ gets at least $k'_2 - 1$, and both
    agents have some non-zero probability of winning the last item. 
    But \pref{prop:alloc} asserts that the number of items agent $1$
    gets shall be integral. So we have a contradiction.
\end{proof}

\begin{lemma} \label{lem:11}
    Suppose $k_1, k_2 \in \Z_{\ge 0}$ satisfy that $k_1 + k_2 = k$, 
    $\frac{k_1}{k_2 + 1} < \frac{B_1}{B_2} < \frac{k_1 + 1}{k_2}$, and
    $\frac{B_1}{B_2} \ne \frac{k_1 + 1}{k_2 + 1}$.
    Then, it must be that one of the agent, say, agent $i$, wins the
    first $k_i$ items and then the other agent wins the remaining items.
\end{lemma}

\begin{proof}
    By \pref{prop:alloc}, agent $1$ gets $k_1$ items and agent $2$ gets
    $k_2$ items in the canonical outcome. Suppose for contradiction
    that the lemma does not hold. Then, by \pref{lem:10}, and our
    assumption that \pref{prop:allocseq} and \pref{prop:pricemon}
    hold for the canonical outcome of $k-1$ item, we can assume
    w.l.o.g.~that agent $1$ wins the first item with price
    $\price{k}{2}$ in the equilibrium; and then agent $2$ wins the next
    $k_2 \ge 1$ items with a non-increasing price sequence $q_1 \ge
    \dots \ge q_{k_2}$ such that the average price is at least
    $\frac{B_2}{k_2 + 1}$; and finally agent $1$ wins the remaining $k_1
    - 1 \ge 1$ items with prices equal $r = B_2 - \sum_{j = 1}^{k_2} q_j
    \le q_{k_2}$, the remaining budget of agent $2$. 

    Let us consider the deviation in which agent $2$ wins the first item
    at price $\price{k}{2}$. By \pref{lem:price}, either agent $2$ is
    indifferent between winning and losing the first item at price
    $\price{k}{2}$, or $W^{(k)}_2(B_2, B_1, p)$ is discontinuous at $p =
    \price{k}{2}$.

    \paragraph{Case 1: Agent $2$ is indifferent}
    In this case, agent $2$ gets $k_2$ items as well in the deviation.
    So the subgame does not fall into the tie-breaking case of
    \pref{prop:alloctie}. Our proof strategy is to first show that the
    price of the first item must be very small. In particular, it is
    smaller than the average price that agent $2$ pays in the canonical
    outcome. Then, we will derive a contradiction by concluding
    agent $2$ has a profitable deviation by wining the first item since
    it is so cheap. We can further divide the situation into three
    cases depending on the allocation sequence after the first round.

    \paragraph{Case 1a: Agent $1$ gets items first after the first round}
    More precisely, after the first round of this deviation, agent $1$
    will win the first $k_1$ items with average price at least
    $\frac{B_1}{k_1 + 1}$ and then agent $2$ wins the remaining $k_2 -
    1$ item with prices equal agent $1$'s remaining budget. 
        
    By \pref{prop:allocseq} of $k-1$ items and our assumed allocation
    sequence after agent $1$ gets the first item, such allocation
    sequence could hold only if the remaining budgets $B_1 -
    \price{k}{2}$ and $B_2$ after the first round satisfy $\frac{B_1 -
    \price{k}{2}}{B_2} > \frac{k_1}{k_2 + 1}$. So we get our first upper
    bound on the price of the first item: $\price{k}{2} < B_1 -
    \frac{k_1}{k_2 + 1} B_2$.
    Via similar reasoning, by \pref{prop:allocseq} and the assumed
    allocation sequence after agent $2$ wins the first time, we get that
    the price of the first item is also upper bounded by $\price{k}{2} <
    B_2 - \frac{k_2}{k_1 + 1} B_1$.
    Combining these two upper bounds we get $\price{k}{2} <
    \frac{B_2}{k_2 + 1}$. 
    
    In this case, we
    argue agent $2$ has a profitable deviation because she could
    have won the first item by bidding $p' \in (\price{k}{2},
    \frac{B_2}{k_2 + 1})$; and then keeps bidding $\price{k}{2}$ until
    agent $1$ wins an item; and finally follows the canonical
    outcome strategy thereafter.
    We let $j^*$ denote the first item that agent $1$ wins in this
    (profitable) deviation. On the one hand, agent $2$ gets $j^* - 1$
    items after $j^*$ rounds in the equilibrium and her remaining
    budget is $B_2 - \sum_{j=1}^{j^* - 1} q_j$ while agent $1$'s
    remaining budget is $B_1 - \price{k}{2}$. On the other hand, agent
    $2$ gets $j^* - 1$ as well after $j^*$ rounds of the deviation and
    agent $1$'s remaining budget is also $B_1 - \price{k}{2}$.
    However, agent $2$'s remaining budget after $j^*$ rounds becomes
    $B_2 - p' - (j^* - 2) \price{k}{2}$. Recall that $q_1
    \ge \dots \ge q_{k_2}$ and $\frac{1}{k_2} \sum_{j = 1}^{k_2} q_j \ge
    \frac{B_2}{k_2 + 1}$. So we have $\sum_{j = 1}^{j^* - 1} q_j \ge
    (j^* - 1) \frac{B_2}{k_2 + 1} > p' + (j^* - 2) \price{k}{2}$. In
    other words, agent $2$'s remaining budget in the deviation is
    strictly larger than that in the equilibrium. By the strict
    monotonicity of $U^{(k - j^*)}_2$ in agent $2$'s budget, agent
    $2$ is strictly better off in the deviation. So we have a
    contradiction.

    \paragraph{Case 1b: Agent $2$ gets items first after the first round}
    In this case, agent $2$ wins the first $k_2 - 1$ items with average
    price at least $\frac{B_2 - \price{k}{2}}{k_2}$ after the first
    round of the deviation, and then agent $1$ wins the remaining items
    with prices equal agent $2$'s remaining budget. 
    Since agent $2$ is indifferent between winning and losing the first
    item, her remaining budget in the two cases shall be the same. So
    agent $1$ wins the remaining item with prices equal $r$ in the
    deviation. Further, agent $1$ pays a total price $\price{k}{2} +
    (k_1 - 1) r$ in the equilibrium and she weakly prefer winning the
    first item than losing it. So we have 
    $$\price{k}{2} \le r \le q_{k_2} \le \dots \le q_1 \enspace.$$
     
    If at least one of the inequalities is strict, then similar to case
    1a, agent $2$ has a profitable deviation by winning the first item
    at price $p' \in (\price{k}{2}, q_1)$, and then bidding
    $\price{k}{2}$ until agent $1$ wins an item, and finally following
    the equilibrium strategy.

    If all the above inequalities hold with equality, then we conclude
    that the prices in the sequential auction are all the same. Let
    $p^*$ denote this fixed price in the auction. Note that $B_2 =
    \sum_{i=1}^{k_2} q_i + r = (k_2 + 1) p^*$. We have $p^* \eqdef
    \frac{B_2}{k_2 + 1}$. Note that by \pref{prop:defaultprice} agent
    $2$ could guarantee getting $k_2$ items paying $\frac{B_1}{k_1 + 1}$
    per item. So
    it must be the case that $\frac{B_1}{k_1 + 1} \ge \frac{B_2}{k_2 +
    1} = p^*$. Therefore, we have $\frac{B_1 - p^*}{k_1} \ge p^*$. By
    \pref{prop:allocseq} and \pref{prop:defaultprice}, if agent $1$
    loses the first item at price $p^*$, then in the subgame it will be
    the case that agent $2$ gets the next $k_2 - 1$ items with average
    price at least $\frac{B_2}{k_2 + 1} = p^*$ and then agent $1$ gets
    the remaining $k_1$ item paying agent $2$'s remaining budget which
    is at most $\frac{B_2}{k_2 + 1} = p^*$. But agent $1$ shall weakly
    prefer winning than losing the first item at $p^*$. So we conclude
    in the deviation where agent $1$ loses the first item at $p^*$, all
    the prices are exactly $p^*$. Thus, agent $1$ has the same utility
    for winning and losing the first item at $p^*$ and $\price{k}{1} =
    p^*$. Hence, we are in the Type II-B Tie-breaking case of
    \pref{prop:allocseq}. We could assume agent $2$ gets the first $k_2$
    items and then agent $1$ gets the remaining $k_1$ items without
    changing the utilities in the equilibrium.
   
    \paragraph{Case 2: $W^{(k)}_2(B_2, B_1, p)$ is discontinuous at $p =
    \price{k}{2}$} 
    
    In this case, let us consider a deviation where
    agent $2$ wins the first item with price $\price{k}{2} + \epsilon$
    and follows the canonical outcome thereafter. By
    \pref{prop:wlmon} of the $k-1$ item case, the utility of
    agent $2$ in this deviation approaches $k_2 v_2 + \frac{k_2}{k_1}
    B_1$ as $\epsilon$ goes to zero, while the utility in the
    equilibrium is $k_2 v_2 + r < k_2 v_2 + \frac{1}{k_2 + 1} B_2$. Yet
    by our assumption $\frac{k_2}{k_1} B_1 \ge \frac{1}{k_1} B_1 >
    \frac{1}{k_2 + 1} B_2$. So the deviation we considered is
    profitable.

    \bigskip

    Summing up all these cases, we either derive contradiction or
    conclude we are in fact in the Type II-A Tie-breaking case.
    Thus, we have proved the lemma.
\end{proof}

Now we are ready to present to proof of \pref{prop:allocseq}.

\begin{proof}[of \pref{prop:allocseq}]
    Suppose $k_1, k_2 \in \Z_{\ge 0}$ satisfy that $k_1 + k_2 = k$ and
    $\frac{k_i}{k_{-i} + 1} < \frac{B_i}{B_{-i}} < \frac{k_i + 1}{k_{-i}
    + 1}$. For the sake of presentation, let us assume w.l.o.g.~that $i
    = 1$.  By \pref{lem:11}, it must be the case that some agent $j$
    wins the first $k_j$ items and then the other agent wins the rest of
    the items. It remains to prove that $j = 1$. Suppose for
    contradiction that $j = 2$. By \pref{prop:defaultprice}, agent $2$
    could guarantee getting the items with prices equal $\frac{B_1}{k_1
    + 1}$. So the remaining budget of agent $2$ is at least $B_2 - k_2
    \frac{B_1}{k_1 + 1} > B_2 - \frac{k_2}{k_2 + 1} B_2 = \frac{B_2}{k_2
    + 1}$. But now agent $i$ must be paying an average price at least
    $\frac{B_2}{k_2 + 1}$, contradicting \pref{lem:11}. 

    Now it remains to analyze the case of $\frac{B_1}{B_2} = \frac{k_1 +
    1}{k_2 + 1}$. Since inductively we have assume the Type II-B
    Tie-breaking of \pref{prop:allocseq} holds in the $k-1$ item case,
    it suffices to show that the critical prices in the first round are
    $\price{k}{1} = \price{k}{2} = p^* \eqdef \frac{B_1}{k_1 + 1} =
    \frac{B_2}{k_2 + 1}$. Further, by symmetry it suffices to show
    $\price{k}{1} = p^*$.
    
    For any price of the first item $p$ that is less than $p^*$, if agent
    $1$ wins the first item at $p$, then in the induced subgame agent
    $1$ has budget strictly greater than $B_1 - p^* = \frac{k_1}{k_2 +
    1} B_2$. So by \pref{prop:defaultprice}, agent $1$'s budget at
    the end shall be strictly greater than $B_1 - p^* - (k_1 - 1)
    \frac{B_2}{k_2 + 1} > \frac{B_1}{k_1 + 1}$. If agent $1$ loses
    the item at $p^*$, on the other hand, then in the induced subgame
    agent $2$'s budget
    will be strictly greater than $B_2 - p^* = \frac{k_2}{k_1 + 1}
    B_1$ according to \pref{prop:defaultprice}. So in this case
    agent $1$, at best, could win $k_1$ items with average price at
    least than $\frac{B_1}{k_1 + 1}$ and have remaining budget at most
    $\frac{B_1}{k_1 + 1}$ in the end. Therefore, we conclude that
    for any price that is strictly less than $p^*$, agent $1$ will
    strictly prefer winning the first item than losing it. 
    Similarly, we could show that for any price that is strictly greater
    than $p^*$, agent $1$ would strictly prefers losing the first item.
    In sum, we have $\price{k}{1} = \price{k}{2} = p^*$.
\end{proof}

\subsubsection{Weakly Declining Prices} \label{sec:pricemon}

In this section, we will prove the prices in the sequential auction is
non-increasing in the number of rounds.
If we are in the tie-breaking cases of \pref{prop:alloctie} or
\pref{prop:allocseq}, then clearly the price sequence is
non-increasing. So it remains to discuss the case without ties.

Let us assume w.l.o.g.~that $\frac{B_1}{k_1 + 1} < \frac{B_2}{k_2 +
1} < \frac{B_1}{k_1}$ for $k_1, k_2 \in \Z_{\ge 0}$ s.t.~$k_1 +
k_2 = k$. Then, by \pref{prop:allocseq} and \pref{prop:defaultprice}
agent $1$ will buy the first $k_1$ items with average price at least
$\frac{B_1}{k_1 + 1}$ and then agent $2$ will win the rest of the
items with price equals agent one's remaining budget, which is at
most $\frac{B_1}{k_1 + 1}$. If the prices are all the same, then it is
clearly non-increasing. So let us further assume the 
prices are not all the same. So the average price agent $1$
pays is strictly greater than $\frac{B_1}{k_1 + 1}$.

If $k_1 = 1$ then the price sequence is clearly non-increasing.
Next, we will assume $k_1 \ge 2$ and agent $1$ wins the first $k_1$
items with prices $q_1, \dots, q_{k_1}$. Further, we will assume for
contradiction that $q_1 < q_2$. By $\ppricemon{k}$, we have $q_2 \ge
\dots \ge q_{k_1}$.

Note that $q_1$ equals the critical price of agent $2$. So by
\pref{lem:price}, either $W^{(k)}_2(B_2, B_1, p)$ is discontinuous
at $p = q_1$, or agent $2$ is indifferent between winning the first
item and losing it at price $q_1$.


In the first case, we have $q_1 = B_2 - \frac{k_2}{k_1} B_1$. Note
that if agent $2$ deviates by winning the first item with price $q_1
+ \epsilon$, then her utility approaches $k_2 v_2 + B_2 - q_1$ as
$\epsilon$ goes to zero. So in the equilibrium, the total price that
agent $2$ pays is at most $q_1$. Now let us consider the second
round, in which agent $1$ has remaining budget $\frac{k_1 +
k_2}{k_1} B_1 - B_2$. Since $q_2 > q_1$ and that agent $2$ pays at
most $q_1$ in total in the equilibrium, we know that agent $2$ could
not be indifferent between winning and losing the second item at
price $q_2$. By \pref{lem:price}, we get that $q_2$ is a discontinuous
point of $W^{(k-1)}_2(\frac{k_1 + k_2}{k_1} B_1 - B_2, B_2)$. Thus,
$q_2 = B_2 - \frac{k_2}{k_1 - 1} \left(\frac{k_1 + k_2}{k_1} B_1 -
B_2\right)$. By our assumption that $q_2 > q_1$, we have $B_2 -
\frac{k_2}{k_1} B_1 < B_2 - \frac{k_2}{k_1 - 1} \left(\frac{k_1 +
k_2}{k_1} B_1 - B_2\right)$, simplifying which we get
$\frac{B_2}{k_2 + 1} > \frac{B_1}{k_1}$. So we have a contradiction
to \pref{prop:alloc}.


Let us move on to the next case that agent $2$ is indifferent
between winning the first item and losing it at price $q_1$. We let
$U_1$ and $U_2$ denote the utilities in the equilibrium of agent $1$
and agent $2$ respectively. We will consider three possible
deviation from the equilibrium path that are summarized in
\pref{tab:deviation}.

The first deviation is when agent $2$ wins the
first item at price $q_1$ and both agents follows the unique
equilibrium of the subgame thereafter. 
In this deviation, the allocation sequence after the first round
must by agent $1$ wins the next $k_1$ items with a non-increasing
price sequence and then agent $2$ wins the remaining $k_2 - 1$ items
paying agent $1$'s remaining budget. Otherwise, by
\pref{prop:allocseq} and \pref{prop:defaultprice}, agent $1$ pays an
average price that is at most $\frac{B_1}{k_1 + 1}$. So agent $1$ is
strictly better off by losing the first item, contradicting our
assumption. We will let $q'_2 \ge \dots \ge q'_{k_1+1}$ denote the
prices at which agent $1$ wins the items. Let $U'_1$ and $U'_2$
denote the utilities in this deviation of agent $1$ and agent $2$
respectively. By our assumption, $U_2 = U'_2$ and $U_1 \ge U'_1$. 

The second deviation is when agent $1$ wins the first item at
price $q_1$ as in the equilibrium, but agent $2$ wins the
second item at price $q_2$, and then both agents follows the unique
equilibrium of the subgame thereafter. Similar to the previous
reasonings, in this subgame it must be
the case that agent $1$ wins the next $k_1 - 1$ items at some
non-increasing price sequence and then agent $2$ wins the remaining
$k_2 - 1$ items paying agent $1$'s remaining budget. We will let
$q''_3 \ge \dots \ge q''_{k_1+2}$ denote the prices at which agent $1$
wins the next $k_1 - 1$ items starting from round $3$, and let
$U''_1$ and $U''_2$ denote the utilities in this deviation of agent
$1$ and agent $2$ respectively. By \pref{lem:price}, both agents
shall weakly prefer winning the second item at price $q_2$ after
losing the first one at $q_1$. So we have $U''_2 \ge U_2$ and $U''_1
\le U_1$.

\begin{table}
    \caption{Summary of the allocation sequences in the equilibrium
    path and in the deviations considered in the proof.}
    \label{tab:deviation}
    \begin{tabular}{cccccc}
        \hline \\ [-1.5ex]
        & Round $1$ & Round $2$ & $\cdots$ & Round $j^*$ & $\cdots$
        \\ [1ex] 
        \hline \\ [-1.5ex]
        Equilibrium & $1$ wins at $q_1$ & $1$ wins at $q_2$ &
        $\cdots$ & $1$ wins at $q_{j^*}$ & $\cdots$ \\ [1ex]
        Deviation $1$ & $2$ wins at $q_1$ & $1$ wins at $q'_2$ &
        $\cdots$ & $1$ wins at $q'_{j^*}$ & $\cdots$ \\ [1ex]
        Deviation $2$ & $1$ wins at $q_1$ & $2$ wins at $q_2$ &
        $\cdots$ & $1$ wins at $q''_{j^*}$ & $\cdots$ \\ [1ex]
        Profitable Deviation & $1$ wins at $q_1$ & $1$ wins at $q_2$
        & $\cdots$ & $2$ wins at $q_1$ & $\cdots$ \\ [1ex]
        \hline
    \end{tabular}
\end{table}

By comparing the utilities of agent $2$ in the unique
equilibrium and in these two deviations, we have $U''_2 \ge U_2 =
U'_2$. Note that $U''_2 = v_2 + U^{(k-2)}_2(B_1 - q_1, B_2 - q_2)$
and $U'_2 = v_2 + U^{(k-2)}_2(B_1 - q'_2, B_2 - q_1)$. So we have
$U^{(k-2)}_2(B_1 - q_1, B_2 - q_2) \ge U^{(k-2)}_2(B_1 - q'_2, B_2 -
q_1)$. Further, we have $B_2 - q_2 < B_2 - q_1$ due to our
assumption that $q_1 < q_2$. So by the monotonicity of
$U^{(k-2)}_2$, we must have $B_1 - q_1 < B_1 - q'_2$, and thus $q_1
> q'_2$. 

Further, by $U_1 \ge U'_1$, and by $U_1 = k_1 v_1 + B_1 -
\sum_{i=1}^{k_1} q_i$ and $U'_1 = k_1 v_1 + B_1 = \sum_{i=2}^{k_1+1}
q'_i$, we have $\sum_{i=1}^{k_1} q_i \le \sum_{i=2}^{k_1+1} q'_i$.
Since $q_1 > q'_2 \ge \dots \ge q'_{k_1+1}$, we conclude that
$\sum_{i=1}^{k_1} q_i < k_1 q_1$. Thus, there exist $3 \le j
\le k_1$ such that $q_j < q_1$. Let $j^*$ denote the smallest such
$j$. 

Now we conclude that agent $2$ has a profitable deviation
because she could let agent $1$ wins the first $j^* - 1$ items at
price $q_1, \dots q_{j^*-1}$, and then wins the next item by bidding
$q_1 - \epsilon > q_{j^*}$. 
The utility of agent $2$ in this deviation will be $U^{(k -
j^*)}_2(B_1 - \sum_{i=1}^{j^*-1} q_i, B_2 - q_1 + \epsilon) > U^{(k
- j^*)}_2(B_1 - \sum_{i=2}^{j^*} q'_i, B_2 - q_1) = U'_2$ due to the
fact that $\sum_{i=1}^{j^*-1} q_i > (j^* - 1) q_1 > \sum_{i=2}^{j^*}
q'_i$.  Further, $U'_2 = U_2$ by our assumption. So this is a
profitable deviation for agent $2$ and we have a contradiction.  

\subsubsection{Monotonicity and Continuity of Utility}
\label{sec:util}

Finally, let us analyze the monotonicity and continuity of $U^{(k)}_i$
with respect to the budgets and prove \pref{prop:utilmon} and
\pref{prop:utilcont}.

\bigskip


    Let us first consider the case when $\frac{B^*_1}{B^*_2} \ne
    \frac{k_1 + 1}{k_2 + 1}$. We will assume w.l.o.g.~that
    $\frac{k_1}{k_2 + 1} < \frac{B^*_1}{B^*_2} < \frac{k_1 + 1}{k_2 +
    1}$. Then, there exists a neighborhood of $(B^*_1, B^*_2)$ such that
    for any budget profile $(B_1, B_2)$ in the neighborhood we have
    $\frac{k_1}{k_2 + 1} < \frac{B_1}{B_2} < \frac{k_1 + 1}{k_2 + 1}$.
    In other words, it is the case that agent $1$ gets the first $k_1$
    items and then agent $2$ gets the remaining items in the
    equilibrium. 

    We will consider the monotonicity and continuity $U^{(k)}_1$
    and $U^{(k)}_2$ as $B_1$ increases. By \pref{prop:allocseq},
    agent $1$ will get items first in the canonical outcome and thus
    agent $1$ will get the first item paying agent $2$'s critical price
    $\price{k}{2}$. Further, consider any sufficiently small $\epsilon
    > 0$ such that $\frac{k_1}{k_2 + 1} < \frac{B^*_1 + \epsilon}{B^*_2}
    < \frac{k_1 + 1}{k_2 + 1}$. We let $p'_2(\epsilon)$ denote the critical
    price of agent $2$ when the budgets are $B^*_1 + \epsilon$ and
    $B^*_2$. We need the following lemmas in our argument.

    \begin{lemma} \label{lem:12}
        $L^{(k)}_2(B^*_2, B^*_1, p)$ is continuous at $p = \price{k}{2}$.
    \end{lemma}
    
    \begin{proof}
        Suppose not. Then, $\price{k}{2}$ is a discontinuous point of
        $L^{(k)}_2(B^*_2, B^*_1, p)$. By \pref{lem:price}, we know
        $W^{(k)}_2(B^*_2, B^*_1, p)$ must be discontinuous at
        $\price{k}{2}$ as well. Thus, it must be the case that
        $\price{k}{2} = B^*_1 - \frac{k_1 - 1}{k_2 + 1} B^*_2$ and
        $\price{k}{2} = B^*_2 - \frac{k_2}{k_1} B^*_1$, which implies
        $\price{k}{2} = \frac{1}{k_2 + 1} B^*_2 = \frac{1}{k_1}
        B^*_1$. So we have $\frac{B^*_1}{B^*_2} = \frac{k_1}{k_2 + 1}$,
        contradicting the assumption in the lemma. 
    \end{proof}

    Given the continuity of $L^{(k)}_2$ at $p = \price{k}{2}$, we
    can upper bound how much the price of the first item increases as
    $B_1$ increases as follows.

    \begin{lemma} \label{lem:13}
        $p'_2(\epsilon) < \price{k}{2} + \epsilon$ for sufficiently
        small $\epsilon > 0$.
    \end{lemma}

    \begin{proof}
        By the monotonicity of $W^{(k)}_2$, for nay $p <
        \price{k}{2} + \epsilon$ we have $W^{(k)}_2(B^*_2, B^*_1 +
        \epsilon, \price{k}{2} + \epsilon) < W^{(k)}_2(B^*_2, B^*_1,
        p)$. So
        \begin{align*}
            W^{(k)}_2(B^*_2, B^*_1 + \epsilon, \price{k}{2} +
            \epsilon) & ~<~ \lim_{p \rightarrow \price{k}{2}+}
            W^{(k)}_2(B^*_2, B^*_1, p) \\
            & ~\le~ \lim_{p \rightarrow \price{k}{2}+}
            L^{(k)}_2(B^*_2, B^*_1, p)
            & \text{(By $\pprice{k}$)} \\
            & ~=~ L^{(k)}_2(B^*_2, B^*_1, \price{k}{2}) 
            & \text{(By \pref{lem:12})} \\
            & ~=~ L^{(k)}_2(B^*_2 + \epsilon, B^*_1, \price{k}{2} +
            \epsilon) \enspace. 
        \end{align*}

        So agent $2$ would prefer losing than winning the first item at
        price $\price{k}{2} + \epsilon$ when the budgets are $B^*_1 +
        \epsilon$ and $B^*_2$. So $p'_2(\epsilon) < \price{k}{2} +
        \epsilon$.
    \end{proof}

    Given \pref{lem:13}, the monotonicity of the agents' utilities in
    $B_1$ follows straightforwardly.

    \begin{lemma} \label{lem:14}
        $U^{(k)}_1(B^*_1 + \epsilon, B^*_2) > U^{(k)}_1(B^*_1,
        B^*_2)$ and $U^{(k)}_2(B^*_2, B^*_1 + \epsilon) \le
        U^{(k)}_2(B^*_2, B^*_1)$
        for sufficiently small $\epsilon > 0$.
    \end{lemma}

    \begin{proof}
        By \pref{lem:13}, we have
        \begin{align*}
            U^{(k)}_1(B^*_1 + \epsilon, B^*_2) & = W^{(k)}_1(B^*_1 +
            \epsilon, B^*_2, p'_2(\epsilon)) > W^{(k)}_1(B^*_1 +
            \epsilon, B^*_2, \price{k}{2} + \epsilon) \\ 
            & = W^{(k)}_1(B^*_1, B^*_2, \price{k}{2}) = U^{(k)}_1(B^*_1,
            B^*_2) 
        \end{align*}
        and 
        \begin{align*}
            U^{(k)}_2(B^*_2, B^*_1 + \epsilon) & = L^{(k)}_2(B^*_2,
            B^*_1 + \epsilon, p'_2(\epsilon)) \le L^{(k)}_2(B^*_2,
            B^*_1 + \epsilon, \price{k}{2} + \epsilon) \\
            & = L^{(k)}_2(B^*_2, B^*_1, \price{k}{2}) = U^{(k)}_2(B^*_2,
            B^*_1) \enspace.
        \end{align*}

        Therefore, we have deduced the desired monotonicity of the
        utilities.
    \end{proof}

    In order to prove continuity, we will first show the continuity of
    agent $2$'s critical price as $B_1$ increases. Since \pref{lem:13}
    already provides us with a lower bound of $p'(\epsilon)$, we only
    need to come up with a lower bound of $p'(\epsilon)$.

    \begin{lemma} \label{lem:15}
        For any sufficiently small $\epsilon > 0$, we have 
        $p'_2(\epsilon) > \price{k}{2} - \frac{k_2}{k_1} \epsilon$.
    \end{lemma}

    \begin{proof}
        First, let us consider the case when $W^{(k)}_2(B^*_2, B^*_1,
        p)$ is continuous at $p = \price{k}{2}$. In this case, by
        \pref{lem:price} we know that $W^{(k)}_2(B^*_2, B^*_1,
        \price{k}{2}) = L^{(k)}_2(B^*_2, B^*_1, \price{k}{2})$. We
        let $p''_2(\epsilon) = \price{k}{2} - \epsilon \frac{(B^*_2 -
        \price{k}{2})}{B^*_1}$. Then, $p''_2(\epsilon) <
        \price{k}{2}$ and it is easy to verify that $B^*_2 -
        p''_2(\epsilon) = \frac{B^*_1 + \epsilon}{B^*_1} (B^*_2 -
        \price{k}{2})$.

        On the one hand, we have
        \begin{align}
            W^{(k)}_2(B^*_2, B^*_1 + \epsilon, p''_2(\epsilon)) & =
            v_2 + U^{(k-1)}_2(B^*_2 - p''_2(\epsilon), B^*_1 + \epsilon)
            \notag \\
            & = v_2 + U^{(k-1)}_2\left(\frac{B^*_1 + \epsilon}{B^*_1}
            (B^*_2 - \price{k}{2}), \frac{B^*_1 + \epsilon}{B^*_1}
            B^*_1\right) \notag \\ 
            & \ge v_2 + U^{(k-1)}_2(B^*_2 - \price{k}{2}, B^*_1) =
            W^{(k)}_2(B^*_2, B^*_1, \price{k}{2}) \enspace.
            \label{eq:2}
        \end{align}

        Here the last inequality holds because when both agents' budgets
        are multiply by the same factor, we will have the same
        allocation sequences and the remaining budget is multiply by the
        same factor.

        On the other hand, by $p''_2(\epsilon) < \price{k}{2}$ we have
        $B^*_1 - \price{k}{2} < B^*_1 + \epsilon - p''_2(\epsilon)$.
        So we have 
        \begin{align} 
            L^{(k)}_2(B^*_2, B^*_1, \price{k}{2}) & = U^{(k-1)}_2(B^*_2,
            B^*_1 - \price{k}{2}) < U^{(k-1)}_2(B^*_2, B^*_1 + \epsilon
            - p''_2(\epsilon)) \notag \\
            & = L^{(k)}_2(B^*_2, B^*_1 + \epsilon, p''_2(\epsilon))
            \enspace.  \label{eq:3}
        \end{align}

        By \pref{eq:2}, \pref{eq:3}, and $L^{(k)}_2(B^*_2, B^*_1,
        \price{k}{2}) = W^{(k)}_2(B^*_2, B^*_1, \price{k}{2})$, we get
        that $L^{(k)}_2(B^*_2, B^*_1 + \epsilon, p''_2(\epsilon)) <
        W^{(k)}_2(B^*_2, B^*_1, p''_2(\epsilon))$. Thus, $p'_2(\epsilon)
        \ge p''_2(\epsilon)$. 
        
        Note that $W^{(k)}_2(B^*_2, B^*_1, \price{k}{2}) =
        L^{(k)}_2(B^*_2, B^*_1, \price{k}{2}) \in [k_2 v_2, (k_2 +
        1) v_2)$ because agent $2$ gets $k_2$ items in the equilibrium.
        So by $W^{(k)}_2(B^*_2, B^*_1, \price{k}{2}) = v_2 +
        U^{(k-1)}_2(B^*_2 - \price{k}{2}, B^*_1)$, we get that
        $U^{(k-1)}_2(B^*_2 - \price{k}{2}, B^*_1) \in
        [(k_2 - 1) v_2, k_2 v_2)$. Therefore, agent $1$ gets $k_1$ items
        and agent $2$ gets $k_2 - 1$ items in the induced subgame after
        agent $2$ wins the first item. So by $\palloc{k}$ we have
        that $\frac{B^*_2 - \price{k}{2}}{B^*_1} < \frac{k_2}{k_1}$
        and therefore $p'_2(\epsilon) \ge p''_2(\epsilon) =
        \price{k}{2} - \frac{\epsilon(B^*_2 - \price{k}{2})}{B^*_1}
        > \price{k}{2} - \epsilon \frac{k_2}{k_1}$.

        Next, we will consider the case when $W^{(k)}_2(B^*_2, B^*_1,
        p)$ is discontinuous at $p = \price{k}{2}$. In this case, we
        know that $\frac{B^*_2 - \price{k}{2}}{B^*_1} =
        \frac{k_2}{k_1}$ and thus $\price{k}{2} = B^*_2 -
        \frac{k_2}{k_1} B^*_1$. Since for any $p < B^*_2 -
        \frac{k_2}{k_1}(B^*_1 + \epsilon)$ we have $\frac{B^*_2 -
        p}{B^*_1} > \frac{k_2}{k_1}$, we get that $W^{(k)}_2(B^*_2,
        B^*_1 + \epsilon, p) \ge (k_2 + 1) v_2 > L^{(k)}_2(B^*_2,
        B^*_1 + \epsilon, p)$ for any $p < B^*_2 - \frac{k_2}{k_1}(B^*_1
        + \epsilon)$. Therefore, we have $p'_2(\epsilon) \ge B^*_2 -
        \frac{k_2}{k_1} (B^*_1 + \epsilon) = \price{k}{2} - \epsilon
        \frac{k_2}{k_1}$.
    \end{proof}

    Now we are ready to show the continuity of $U^{(k)}_1$ and
    $U^{(k)}_2$ in $B_1$.

    \begin{lemma} \label{lem:16}
        We have
        $$\lim_{\epsilon \rightarrow 0^+} U^{(k)}_1(B^*_1 + \epsilon,
        B^*_2) = U^{(k)}_1(B^*_1, B^*_2) \quad,\quad
        \lim_{\epsilon \rightarrow 0^+} U^{(k)}_2(B^*_2, B^*_1 +
        \epsilon) = U^{(k)}_2(B^*_2, B^*_1) \enspace.$$
    \end{lemma}

    \begin{proof}
        Note that
        \begin{align}
            \lim_{\epsilon \rightarrow 0^+} U^{(k)}_1(B^*_1 +
            \epsilon, B^*_2) & = \lim_{\epsilon \rightarrow 0^+}
            \left(v_1 + U^{(k-1)}_1(B^*_1 + \epsilon - p'_2(\epsilon),
            B^*_2)\right) \label{eq:8} \\
            \lim_{\epsilon \rightarrow 0^+} U^{(k)}_2(B^*_2, B^*_1 +
            \epsilon) & = \lim_{\epsilon \rightarrow 0^+}
            U^{(k-1)}_2(B^*_2, B^*_1 + \epsilon - p'_2(\epsilon))
            \label{eq:9}
        \end{align}

        Note that we have inductively assume that
        $U^{(k-1)}_1$ and $U^{(k-1)}_2$ are continuous at $B_1 = B^*_1 -
        \price{k}{2}$. Further, by \pref{lem:13} and \pref{lem:15}, we
        know that $\lim_{\epsilon \rightarrow 0^+} (B^*_1 + \epsilon -
        p'_2(\epsilon)) = B^*_1 - \price{k}{2}$. So from \pref{eq:8}
        we have $\lim_{\epsilon \rightarrow 0^+} U^{(k)}_1(B^*_1 +
        \epsilon, B^*_2) = v_1 + U^{(k-1)}_1(B^*_1 - p'_2(\epsilon),
        B^*_2) = U^{(k)}_1(B^*_1, B^*_2)$ and from \pref{eq:9} we have
        $\lim_{\epsilon \rightarrow 0^+} U^{(k)}_2(B^*_2, B^*_1 +
        \epsilon) = U^{(k-1)}_2(B^*_2, B^*_1 - p'_2(\epsilon)) =
        U^{(k)}_2(B^*_2, B^*_1)$.
    \end{proof}

    Similarly, we could show the desired monotonicity and continuity of
    $U^{(k)}_1$ and $U^{(k)}_2$ as $B_1$ decreases. We will omit
    the tedious calculation in this paper.

    Next, let us move on to how $U^{(k)}_1$ and $U^{(k)}_2$ behaves
    as $B_2$ changes. Again, we will only consider the case that $B_2$
    increases as the opposite case is very similar.

    \begin{lemma} \label{lem:17}
        Suppose a budget profile $(B_1, B_2)$ satisfies
        $\frac{k'_1}{k'_2 + 1} < \frac{B_1}{B_2} < \frac{k'_1 + 1}{k'_2
        + 1}$ for $k'_1 + k'_2 = k - 1$. Then, for sufficiently
        small $\epsilon > 0$, we have $U^{(k-1)}_2(B_2 + \epsilon, B_1)
        \ge U^{(k-1)}_2(B_2, B_1) + \epsilon$.
    \end{lemma}

    \begin{proof}
        Since $\frac{k'_1}{k'_2 + 1} < \frac{B_1}{B_2} < \frac{k'_1 +
        1}{k'_2 + 1}$, in the canonical outcome of a sequential
        auction with $k-1$ item and budget profile $(B_1, B_2)$, agent
        $1$ will get the first $k'_1$ items and agent $2$ gets the
        remaining $k'_2$ items. Further, for sufficiently small
        $\epsilon$, we have the same allocation sequence in the
        canonical outcome when the budget profile is $(B_1, B_2 +
        \epsilon)$. 
        
        By \pref{prop:utilmon} of the $k-1$ item cases, we have
        $U^{(k-1)}_1(B_1, B_2) \ge U^{(k-1)}_1(B_1, B_2 + \epsilon)$. So
        the total price that agent $1$ pays becomes higher or remains
        the same as agent $2$'s budget changes from $B_2$ to $B_2 +
        \epsilon$. Note that in the equilibrium allocation sequences of
        both cases, agent $2$ pays agent $1$'s remaining budget for each
        item. So the total price that agent $2$ pays decreases and
        remains the same as her budget increases from $B_2$ to $B_2 +
        \epsilon$.  Thus, we have $U^{(k-1)}_2(B_2 + \epsilon, B_1) \ge
        U^{(k-1)}_2(B_2, B_1) + \epsilon$.
    \end{proof}

    \begin{lemma} \label{lem:18} 
        $\frac{k_1}{k_2} \le \frac{B^*_1}{B^*_2 - \price{k}{2}} <
        \frac{k_1 + 1}{k_2}$.
    \end{lemma}

    \begin{proof}
        By $W^{(k)}_2(B^*_2, B^*_1, \price{k}{2}) \ge
        L^{(k)}_2(B^*_2, B^*_1, \price{k}{2}) = U^{(k)}_2(B^*_2,
        B^*_1)$ and that agent $2$ gets $k_2$ items in equilibrium, we
        get that agent $2$ gets at least $k_2 - 1$ items in the induced
        subgame after winning the first item at price $\price{k}{2}$.
        So we have $\frac{k_1}{k_2} \le \frac{B^*_1}{B^*_2 -
        \price{k}{2}} < \frac{k_1 + 1}{k_2 - 1}$.

        Next, we will show that $\frac{B^*_1}{B^*_2 - \price{k}{2}} <
        \frac{k_1 + 1}{k_2}$. Suppose not. Then, by $\pallocseq{k}$ and
        $\pallocseqtie{k}$ we have that agent $2$'s remaining budget at
        the end of the induced subgame after she wins the first round is
        at most $\frac{B^*_1}{k_1 + 1}$. On the other hand, by
        $\pallocseq{k}$ we know that in the equilibrium of the
        $k$-item sequential auction, agent $2$'s remaining budget is
        at least $B^*_2 - k_2 \frac{B^*_1}{k_1 + 1} > \frac{k_2 + 1}{k_1
        + 1} B^*_1 - k_2 \frac{B^*_1}{k_1 + 1} = \frac{B^*_1}{k_1 + 1}$.
        So we deduce that $W^{(k)}_2(B^*_2, B^*_1, \price{k}{2}) <
        L^{(k)}_2(B^*_2, B^*_1, \price{k}{2})$, contradicting
        \pref{lem:price}.
    \end{proof}

    Further, since that agent $1$ gets the first $k_1$ item in the
    equilibrium, we know that agent $1$ gets the first $k_1 - 1$ in the
    induced subgame after she wins the first item at price
    $\price{k}{2}$. So by \pref{prop:allocseq} of the $k-1$ item cases,
    we have the following.

    \begin{lemma} \label{lem:19}
        $\frac{k_1 - 1}{k_2 + 1} < \frac{B^*_1 - \price{k}{2}}{B^*_2} <
        \frac{k_1}{k_2 + 1}$.
    \end{lemma}

    \begin{lemma} \label{lem:20}
        $U^{(k)}_2(B^*_2 + \epsilon, B^*_1) \ge U^{(k)}_2(B^*_2,
        B^*_1) + \epsilon$ for sufficiently small $\epsilon > 0$.
    \end{lemma}

    \begin{proof}
        If $\frac{k_1}{k_2} = \frac{B^*_1}{B^*_2 - \price{k}{2}}$,
        then $W^{(k)}_2(B^*_2, B^*_1, p)$ is discontinuous at $p =
        \price{k}{2}$.  In this case, agent $2$ could have bid
        $\price{k}{2} + \epsilon$ when her budget is $B^*_2 + \epsilon$.
        If she wins the first round, then by $B^*_2 + \epsilon -
        (\price{k}{2} + \epsilon) = B^*_2 - \price{k}{2} =
        \frac{k_2}{k_1} B^*_1$, she has some chance of winning $k_2 + 1$
        items. So the utility would greater than $U^{(k)}_2(B^*_2,
        B^*_1) + \epsilon$, where agent $2$ only gets $k_2$ items. If
        she loses the first round, then her utility is
        $U^{(k-1)}_2(B^*_2 + \epsilon, B^*_1 - \price{k}{2} -
        \epsilon)$. By \pref{lem:17} and the monotonicity of $U^{(k-1)}_2$
        in $B_1$, this utility is greater or equal to $U^{(k-1)}_2(B^*_2,
        B^*_1 - \price{k}{2}) + \epsilon = U^{(k)}_2(B^*_2, B^*_1) +
        \epsilon$.
        If $\frac{k_1}{k_2} \ne \frac{B^*_1}{B^*_2 - \price{k}{2}}$,
        then by \pref{lem:18} we have $\frac{k_1}{k_2} <
        \frac{B^*_1}{B^*_2 - \price{k}{2}} < \frac{k_1 + 1}{k_2}$.
        Moreover, for sufficiently small $\epsilon > 0$, we have
        $\frac{k_1}{k_2} < \frac{B^*_1}{B^*_2 + \epsilon -
        \price{k}{2}} < \frac{k_1 + 1}{k_2}$. Thus, by letting $k'_1 =
        k_1$ and $k'_2 = k_2 - 1$ in \pref{lem:17}, we get that 
        \begin{align}
            W^{(k)}_2(B^*_2 + \epsilon, B^*_1, \price{k}{2}) & = v_2 +
            U^{(k-1)}_2(B^*_2 + \epsilon - \price{k}{2}, B^*_1) &
            \text{(Definition of $W^{(k)}_2$)} \notag \\ 
            & \ge v_2 + U^{(k-1)}_2(B^*_2 - \price{k}{2}, B^*_1) +
            \epsilon & \text{(\pref{lem:17})} \notag \\ 
            & = W^{(k)}_2(B^*_2, B^*_1, \price{k}{2}) + \epsilon &
            \text{(Definition of $W^{(k)}_2$)} \notag \\
            & \ge L^{(k)}_2(B^*_2, B^*_1, \price{k}{2}) + \epsilon
            & \text{($\pprice{k}$)} \notag \\
            & = U^{(k)}_2(B^*_2, B^*_1) + \epsilon \enspace.
            \label{eq:4}
        \end{align}

        Further, by \pref{lem:19} we have $\frac{k_1 - 1}{k_2 + 1} <
        \frac{B^*_1 - \price{k}{2}}{B^*_2} < \frac{k_1}{k_2 + 1}$. So for
        sufficiently small $\epsilon > 0$, we have $\frac{k_1 - 1}{k_2 +
        1} < \frac{B^*_1 - \price{k}{2}}{B^*_2 + \epsilon} <
        \frac{k_1}{k_2 + 1}$. Therefore, by letting $k'_1 = k_1 - 1$ and
        $k'_2 = k_2$ in \pref{lem:17}, we have that
        \begin{align}
            L^{(k)}_2(B^*_2 + \epsilon, B^*_1, \price{k}{2}) & = 
            U^{(k-1)}_2(B^*_2 + \epsilon, B^*_1 - \price{k}{2}) \ge 
            U^{(k-1)}_2(B^*_2, B^*_1 - \price{k}{2}) + \epsilon \notag \\
            & = L^{(k)}_2(B^*_2, B^*_1, \price{k}{2}) + \epsilon
            = U^{(k)}_2(B^*_2, B^*_1) + \epsilon \enspace.
            \label{eq:5}
        \end{align}

        By \pref{eq:4} and \pref{eq:5}, agent $2$ could have bid
        $\price{k}{2}$ and guaranteed at least $U^{(k)}_2(B^*_2,
        B^*_1) + \epsilon$ utility when her budget is $B^*_2 + \epsilon$
        for sufficiently small $\epsilon$. Thus, $U^{(k)}_2(B^*_2 +
        \epsilon, B^*_1) \ge U^{(k)}_2(B^*_2, B^*_1) + \epsilon$.
    \end{proof}

    \begin{lemma} \label{lem:21}
        $U^{(k)}_1(B^*_1, B^*_2 + \epsilon) \le U^{(k)}_1(B^*_1,
        B^*_2)$ for sufficiently small $\epsilon > 0$.
    \end{lemma}

    \begin{proof}
        Let us consider the equilibrium allocation sequences when the
        budget profiles are $(B^*_1, B^*_2)$ and $(B^*_1, B^*_2 +
        \epsilon)$ for sufficiently small $\epsilon > 0$. By our
        assumption, in both cases agent $1$ will gets the first $k_1$
        items and then agent $2$ will get the remaining $k_2$ items
        paying agent $1$'s remaining budget for each item.

        By \pref{lem:20}, we have $U^{(k)}_2(B^*_2 + \epsilon, B^*_1)
        \ge U^{(k)}_2(B^*_2, B^*_1) + \epsilon$. Thus, the total price
        that agent $2$ pays when her budget is $B^*_2 + \epsilon$ is
        lower than that when her budget is $B^*_2$. So we conclude that
        the remaining budget of agent $1$ when agent $2$'s budget is
        $B^*_2 + \epsilon$ is smaller than that when agent $2$'s budget
        is $B^*_2$. In other words, $U^{(k)}_1(B^*_1, B^*_2 +
        \epsilon) \le U^{(k)}_1(B^*_1, B^*_2)$.
    \end{proof}

    By \pref{lem:20} and \pref{lem:21}, we have shown the desired
    monotonicity of $U^{(k)}_1$ and $U^{(k)}_2$ in $B_2$. It remains
    to show the utility functions are continuous in $B_2$ at point
    $(B^*_1, B^*_2)$. 
    We will let $p'_2(\epsilon)$ denote the critical price of agent $2$
    when the budgets are $B^*_1$ and $B^*_2 +
    \epsilon$. In other words, $p'_2(\epsilon)$ is the price that agent
    $1$ pays in the first round. 

    \begin{lemma} \label{lem:22}
        For sufficiently small $\epsilon > 0$, we have $p'_2(\epsilon)
        \le \price{k}{2} + \epsilon$. 
    \end{lemma}

    \begin{proof}
        Consider a price $\price{k}{2} + \epsilon + \epsilon'$ for
        sufficiently small $\epsilon' > 0$. By the definition of
        $W^{(k)}_2$, we have 
        \begin{align}
            W^{(k)}_2(B^*_2 + \epsilon, B^*_1, \price{k}{2} + \epsilon +
            \epsilon') & ~=~ v_2 + U^{(k-1)}_2(B^*_2 - \price{k}{2} -
            \epsilon', B^*_1) \notag \\
            & ~=~ W^{(k)}_2(B^*_2, B^*_1, \price{k}{2} + \epsilon')
            \enspace. \label{eq:6} 
        \end{align}

        Further, by the monotonicity of $U^{(k-1)}_2$, we have 
        \begin{align}
            L^{(k)}_2(B^*_2 + \epsilon, B^*_1, \price{k}{2} +
            \epsilon + \epsilon') & = U^{(k-1)}_2(B^*_2 + \epsilon,
            B^*_1 - \price{k}{2} - \epsilon - \epsilon') \notag \\
            & > U^{(k-1)}_2(B^*_2, B^*_1 - \price{k}{2} - \epsilon')
            \notag \\
            & = L^{(k)}_2(B^*_2, B^*_1, \price{k}{2} + \epsilon')
            \enspace.  \label{eq:7}
        \end{align} 
        
        Finally, by definition of $\price{k}{2}$, for any $p >
        \price{k}{2}$, we have $W^{(k)}_2(B^*_2, B^*_1, p) <
        L^{(k)}_2(B^*_2, B^*_1, p)$. So by \pref{eq:6}, \pref{eq:7},
        we have $W^{(k)}_2(B^*_2 + \epsilon, B^*_1, \price{k}{2} +
        \epsilon + \epsilon') < L^{(k)}_2(B^*_2 + \epsilon, B^*_1,
        \price{k}{2} + \epsilon + \epsilon')$. Note that this hold for
        any $\epsilon' > 0$. So by the definition of critical prices,
        $p'_2(\epsilon) \le \price{k}{2} + \epsilon$.
    \end{proof}

    \begin{lemma} \label{lem:23}
        We have
        $$\lim_{\epsilon \rightarrow 0^+} U^{(k)}_1(B^*_1, B^*_2 +
        \epsilon) = U^{(k)}_1(B^*_1, B^*_2) \quad\text{and}\quad
        \lim_{\epsilon \rightarrow 0^+} U^{(k)}_2(B^*_2 + \epsilon,
        B^*_1) = U^{(k)}_2(B^*_2, B^*_1) \enspace.$$
    \end{lemma}

    \begin{proof}
        On the one hand, by \pref{lem:20} we have $U^{(k)}_2(B^*_2 +
        \epsilon, B^*_1) \ge U^{(k)}_2(B^*_2, B^*_1) + \epsilon$,
        which goes to $U^{(k)}_2(B^*_2, B^*_1)$ as $\epsilon$ goes to
        zero. On the other hand, by \pref{lem:22}, we have
        $U^{(k)}_2(B^*_2 + \epsilon, B^*_1) = U^{(k-1)}_2(B^*_2 +
        \epsilon, B^*_1 - p'_2(\epsilon)) \le U^{(k-1)}_2(B^*_2 +
        \epsilon, B^*_1 - \price{k}{2} - \epsilon)$, which also goes to
        $U^{(k-1)}_2(B^*_2, B^*_1 - \price{k}{2}) = U^{(k)}_2(B^*_2,
        B^*_1)$ as $\epsilon$ goes to
        zero due to the continuity of $U^{(k)}_2$. So we have
        $\lim_{\epsilon \rightarrow 0^+} U^{(k)}_2(B^*_2 + \epsilon,
        B^*_1) = U^{(k)}_2(B^*_2, B^*_1)$.

        Note that agent $2$ will get the last $k_2$ items paying agent
        $1$'s remaining budget for each item. So the continuity of agent
        $2$'s utility implies that the remaining budget of agent $1$ and
        and the utility of agent $1$ are continuous in $B_2$ at point
        $(B^*_1, B^*_2)$. 
    \end{proof}

    Finally, it remains to prove the continuity and monotonicity at
    $\frac{B^*_1}{B^*_2} = \frac{k_1 + 1}{k_2 + 1}$. In fact, we can use
    the same arguments as above, except that we will assume the
    allocation sequence is agent $1$ gets item first when we consider
    $B_1$ approaches $B^*_1$ from below or $B_2$ approaches $B^*_2$ from
    above, and assume agent $2$ getting items first for the other two
    cases. 

\begin{proof}
    Since we have proved \pref{lem:14}, \pref{lem:16}, \pref{lem:20},
    \pref{lem:21},  \pref{lem:23}, it remains to verify the
    boundary conditions. We will analyze how $U^{(k)}_1(B_1, B_2)$
    behaves as $B_2$ approaches $\frac{k_2}{k_1} B_1$ from below. The
    proof of the other cases are very similar and we will omit them in
    this extended abstract.

    Suppose $B_2 = \frac{k_2 + 1}{k_1} B_1 - \epsilon$ for 
    sufficiently small $\epsilon > 0$. Then, by \pref{prop:allocseq},
    agent $1$ will get the first $k_1$ items in the canonical
    outcome and then agent $2$ will get the remaining $k_2$ items
    paying agent $1$'s remaining budget. In particular, agent $1$ will
    win the first item paying agent $2$'s critical price $\price{k}{2}$. 

    Next, we will let $p^* = \frac{B_1}{k_1}$ for the sake of expedition
    and argue $\price{k}{2} \ge p^* - \epsilon$. If agent $2$ wins the
    first item at $p^* - \epsilon$, then her remaining budget becomes
    $\frac{k_2}{k_1} B_1$. So by \pref{prop:alloctie}, agent $2$ will
    have non-zero probability of winning $k_2 + 1$ items. Hence, agent
    $2$ shall strictly prefers winning the first item at $p^* -
    \epsilon$. So we have $\price{k}{2} \ge p^* - \epsilon$.

    Therefore, we have $U^{(k)}_1(B_1, B_2) = U^{(k-1)}(B_1 -
    \price{k}{2}, B_2) \le U^{(k-1)}_1(B_1 - p^* + \epsilon, B_2)$. If
    we let $\epsilon$ goes to zero, the limit of the right-hand-size
    goes to $k_1 v_1$ by our inductive hypothesis. So 
    $$\lim_{B_2 \rightarrow \frac{k_2 + 1}{k_1} B_1^-} U^{(k)}_1(B_1,
    B_2) \le k_1 v_1 \enspace.$$

    Further, $U^{(k)}_1(B_1, B_2) \ge k_1 v_1$ for any $\epsilon > 0$.
    So the above holds with equality.

    \bigskip

    Finally, let us consider the utility of agent $i$ when
    $\frac{B_i}{B_{-i}} = \frac{k_i}{k_{-i} + 1}$. We will assume
    w.l.o.g.~that $i=1$ for the sake of presentation. By
    \pref{prop:alloctie}, both agents will keep bidding $p^* \eqdef
    \frac{B_1}{k_1} = \frac{B_2}{k_2 + 1}$ until one of the agent runs
    out of her budget. So agent $1$ will get $k_1 - 1$ items for sure,
    and with probability $\phi(k_2, k_1 - 1)$, agent $1$ will get an
    extra item. For $n = 1, \dots, k_1 - 1$, with probability $\phi(n,
    k_2) - \phi(n - 1, k_2) = {k_2 + n \choose n}
    \left(\frac{1}{2}\right)^{k_2 + n + 1}$, agent $1$ gets exactly $n$
    items before agent $2$ gets $k_2 + 1$ items. So the utility for
    agent $i$ is
    $$U^{(k)}_1(B_1, B_2) = (k_1 - 1) v_1 + \phi(k_2, k_1 - 1) v_1 +
    \sum_{n=1}^{k_1 - 1} {k_2 + n \choose n}
    \left(\frac{1}{2}\right)^{k_2 + n + 1} (B_1 - np^*) \enspace.$$

    Further, 
    \begin{eqnarray*}
        & & \sum_{n=1}^{k_1 - 1} {k_2 + n \choose n}
        \left(\frac{1}{2}\right)^{k_2 + n + 1} (B_1 - np^*) \\ 
        & = & \phi(k_1 - 1, k_2) B_1 - (k_2 + 1) p^* \sum_{n=1}^{k_1 -
        1} {k_2 + n \choose n - 1} \left(\frac{1}{2}\right)^{k_2 + n +
        1} \\
        & = & \phi(k_1 - 1, k_2) B_1 - (k_2 + 1) p^* \phi(k_1 - 2, k_2 +
        1) \enspace.
    \end{eqnarray*}

    Since, $p^* = \frac{B_2}{k_2 + 1}$ by our assumption, we have proved
    the asserted utility.
\end{proof}

\section{Semi-Trembling-Hand-Perfection of the Canonical Outcome} \label{app:unique}

In this section, we will prove that the canonical outcome is the only
stable equilibrium under the refinement of semi-trembling-hand-perfection.

\begin{lemma} \label{lem:refinement1}
    For $i = 1, 2$, suppose $\price{k}{i} > \price{k}{-i}$. Then, the
    canonical outcome is the unique semi-trembling-hand-perfect and
    subgame-perfect equilibrium.
\end{lemma}

In the following discussion, let us assume w.l.o.g.~that $i = 1$ in
\pref{lem:refinement1} for the sake of presentation. In order to prove
\pref{lem:refinement1}, we will first show a few lemmas. The first lemma
clarifies the potential equilibrium strategies we need to consider.

\begin{lemma} \label{lem:refinement2}
    Suppose $\price{k}{1} > \price{k}{2}$ and the
    agents follow the canonical outcome in the subgames of $k-1$
    items. Then, the only candidate equilibrium strategy for the first
    round is agent $1$ bidding $p+$ and agent $2$ bidding $p$ for
    $\price{k}{1} \ge p \ge \price{k}{2}$.
\end{lemma}

\begin{proof}
    First, we note that one of the agents bidding strictly greater than
    the other ($p+$ is not considered greater than $p$)
    cannot be an equilibrium because by the monotonicity of $W^{(k)}_1$
    and $W^{(k)}_2$, the winner will make a lower bid in order to get
    the first item with a lower price. Further, both agents bidding
    strictly greater than $\price{k}{1}$ cannot be an equilibrium
    because the winner (in case of tie, both agents) will prefer losing
    the first item at such a high price. Moreover, both agents bidding
    strictly smaller then $\price{k}{2}$ cannot be an equilibrium
    either, because the loser (in case of tie, both agents) will prefer
    bid slightly higher and wins the first item. Finally, both agents
    bidding $p$ for $\price{k}{1} \ge p \ge \price{k}{2}$, that is,
    agent $1$ did not use the privilege of bidding $p+$, cannot be an
    equilibrium, because either agent $1$ would strictly prefer bidding
    $p+$ (if $p < \price{k}{1}$) or agent $2$ would strictly prefer
    underbids and losing the item (if $p > \price{k}{2}$). In sum, the
    only candidate equilibrium strategy for the first round is agent $1$
    bidding $p+$ and agent $2$ bidding $p$ for $\price{k}{1} \ge p \ge
    \price{k}{2}$.
\end{proof}

\begin{lemma} \label{lem:refinement5}
    Then, bidding $p > \price{k}{2}$ in the first round (and follows
    the unique equilibrium in the induced subgame) is weakly dominated
    for agent $2$.
\end{lemma}

\begin{proof}
    Let us consider the alternative strategy of bidding $\price{k}{2} +
    \epsilon < p$ for sufficiently small $\epsilon > 0$. We will show
    this strategy weakly dominates bidding $p$ for agent $2$.
    
    If agent $1$ bids $\hat{p} > p$, then both strategies lose the
    first item at $\hat{p}$ and therefore yield the same payoff. 
    
    If agent $1$ bids $\hat{p} = p$ or $p+$, then by $p >
    \price{k}{2}$ we have $L^{(k)}_2(B_2, B_1, p) > W^{(k)}_2(B_2, B_1,
    p)$. So the utility of bidding $p$ is at most $L^{(k)}_2(B_2, B_1,
    p)$, which equals the utility of bidding $\price{k}{2} + \epsilon$
    and losing the first item. 
    
    If agent $1$ bids $\hat{p}$ s.t.~$\price{k}{2} + \epsilon \le
    \hat{p} < p$, then the utility of bidding $p$ is $W^{(k)}_2(B_2,
    B_1, p)$.  By the monotonicity of $W^{(k)}_2$, this is less than
    $W^{(k)}_2(B_2, B_1, \hat{p})$. Further, by $\hat{p} > \price{k}{2}
    + \epsilon \ge \price{k}{2}$, we get that $W^{(k)}_2(B_2, B_1,
    \hat{p}) \le L^{(k)}_2(B_2, B_1, \hat{p})$. So bidding
    $\price{k}{2} + \epsilon$ and losing the first item at $\hat{p}$ is
    strictly better.  

    Finally , if agent $1$ bids $\hat{p} < \price{k}{2} + \epsilon$,
    then the both strategies wins the first item. So bidding
    $\price{k}{2}$ is strictly better for that $W^{(k)}_2$ is decreasing
    as the price increases.
\end{proof}

\begin{lemma} \label{lem:refinement3}
    Suppose $\price{k}{1} > \price{k}{2}$ and the agents follow the
    canonical outcome in the subgames of $k-1$ items. Then, 
    for any $p$ s.t.~$\price{k}{1} \ge p > \price{k}{2}$, agent $1$
    bidding $p+$ and agent $2$ bidding $p$ is not a
    semi-trembling-hand-perfect equilibrium.
\end{lemma}

\begin{proof}
    Consider any sequence $\{\sigma_j\}_j$ of completely mixed
    strategies of agent $1$ that converges to bidding $p+$. We will
    argue bidding strictly greater than $\price{k}{2}$ is sub-optimal
    for agent $2$ when agent $1$ use strategy $\sigma_j$ for any $j$
    because it is weakly dominated. Therefore, the best responses of
    $\{\sigma_j\}_j$ cannot converges to bidding $p$ since $p >
    \price{k}{2}$. So it is not a semi-trembling-hand-perfect
    equilibrium.
\end{proof}

\begin{lemma} \label{lem:refinement4}
    Suppose $\price{k}{1} > \price{k}{2}$ and the agents follow the
    canonical outcome in the subgames of $k-1$ items. Then, agent
    $1$ bidding $\price{k}{2}+$ and agent $2$ bidding $\price{k}{2}$ is
    a semi-trembling-hand-perfect equilibrium.
\end{lemma}

\begin{proof}
    Let us first consider the stability of agent $2$'s strategy.
    Consider the following sequence $\{\sigma_j\}_j$ of completely mixed
    strategies of agent $1$ that converges to bidding $\price{k}{2}+$.
    Let
    $$\alpha_j = W^{(k)}_2(B_2, B_1, \price{k}{2}(1 - 2^{-j})) -
    W^{(k)}_2(B_2, B_1, 0)$$
    denote the gain of winning the first item for free instead of
    $\price{k}{2}(1 - 2^{-j-1})$ for agent $2$.
    Let $$\beta_j = W^{(k)}_2(B_2, B_1, \price{k}{2}(1 - 2^{-j-1})) -
    L^{(k)}_2(B_2, B_1, \price{k}{2}(1 - 2^{-j-1}))$$
    denote the gain of winning the first item rather than losing it for
    agent $2$ when the price is $\price{k}{2}(1 - 2^{-j-1})$. We will
    let $\gamma_j = \min\{1, \frac{\beta_j}{\alpha_j}\}$ and define
    $\sigma_j$ such that the probability density of agent $1$ bidding
    $p$ when she uses $\sigma_j$ is:
    $$f_{\sigma_j}(p) = \left\{\begin{aligned}
        & \frac{1}{2^j p_2} & & \text{ , if $\abs{p - \price{k}{2}}
        \le \frac{p_2}{2^j}$} \\
        & \frac{\gamma_j}{2^{2j+2} p_2} & & \text{ , if $\abs{p -
        \price{k}{2}} > \frac{p_2}{2^j}$} 
    \end{aligned}\right. \enspace,$$
    and we will choose the probability of bidding $\price{p}{2}+$
    properly such that the probability sum up to $1$. It is easy to
    verify this sequence of completely mixed strategies converges to
    bidding $\price{p}{2}$. 
    
    Further, by \pref{lem:refinement5} we get that the best response
    must be bids smaller or equal to $\price{p}{2}$. 
    
    Finally, we claim any bid $p$ that is smaller than
    $\price{p}{2}(1 - 2^{-j})$ is strictly worse off comparing to
    bidding $\price{p}{2}(1 - 2^{-j-1})$. When agent $1$ bids above
    $\price{k}{2}(1 - 2^{-j-1})$, both strategy yields the same payoff.
    When agent $1$ bids between $\price{k}{2}(1 - 2^{-j})$ and
    $\price{k}{2}(1 - 2^{-j-1})$, which happens with prabability
    $2^{-2j-1}$ by our choice of $\sigma_j$, bidding $\price{k}{2}(1 -
    2^{-j-1})$ is better by at least $\beta_j$. When agent $1$ bids
    below $\price{k}{2}(1 - 2^{-j})$, which happens with probability at
    most $\gamma_j 2^{-2j-2}$ by our choice of $\sigma_j$, bidding
    $p$ could be better off by at most $\alpha_j$. So by $\gamma_j \le
    \frac{\beta_j}{\alpha_j}$, we get that bidding $\price{k}{2}(1 -
    2^{-j-1})$ is better for agent $2$. Therefore, we get that any best
    response bid to $\sigma_j$ must be at least $\price{k}{2}(1 -
    2^{-j-1})$.

    Summing up the above upper and lower bounds on the best response
    bids of agent $2$, we get that the best responses of $\sigma_j$
    converges to bidding $\price{p}{2}$ as $j$ increases.

    The stability of agent $1$'s strategy can be proved similarly. So
    the canonical outcome is semi-trembling-hand-perfect.
\end{proof}

Summarizing \pref{lem:refinement2}, \pref{lem:refinement5},
\pref{lem:refinement3}, and \pref{lem:refinement4}, we have proved
\pref{lem:refinement1}. Via simlar analysis, we can show the canonical
outcome is ``stable'' as well when the critical prices are the same
in the first round. We will omit the details here.

\begin{lemma} \label{lem:refinement6}
    Suppose $\price{k}{1} = \price{k}{2}$. Then, the canonical
    outcome is the unique semi-trembling-hand-perfect and
    subgame-perfect equilibrium.
\end{lemma}

\end{document}